\documentclass[envcountsame,runningheads]{llncs}
\usepackage[T1]{fontenc}
\usepackage{graphicx}
\usepackage{tabularray}

\usepackage[all]{xy}

\usepackage[mathscr]{eucal}
\usepackage{amssymb}
\usepackage{enumerate}
\usepackage{url} 
\usepackage{multicol}
\usepackage{multirow}
\usepackage{enumitem}

\usepackage{amsmath}
\usepackage{marginnote}
\usepackage{comment}

\usepackage{pgf,tikz}

\usetikzlibrary{matrix}
\usetikzlibrary{calc,positioning,shapes,arrows.meta,fit}

\tikzset{%
 shaded/.style={draw, shape=circle, fill=black!35, inner sep=1.4pt},
 unshaded/.style={draw, shape=circle, fill=white, inner sep=1.4pt},
 quasi/.style={draw, shape=rectangle, rounded corners=3pt, fill=white, inner sep=2.5pt, minimum height=14.5pt},
 blob/.style={draw, shape=rectangle, rounded corners=12pt, thin, densely dotted},
 arrow/.style={->, thin, >=latex, shorten >=2.5pt, shorten <=2.5pt},
 order/.style={thin},
 curvy/.style={thin, looseness=1.2, bend angle=70},
 fatcurvy/.style={thin, looseness=1.7, bend angle=75},
 label/.style={shape=rectangle, inner sep=6pt},
 auto}

\font\bmi=cmmi8 scaled 1440
\newcommand{\powerset}{\raise.6ex\hbox{\bmi\char'175 }}
\newcommand{\UpE}{\mathsf{Up}(\mathbf{E})}

\newcommand{\tw}{{\sim}}

\newcommand{\PSI}{\mathsf{psi}}
\newcommand{\A}{\mathbf{A}}

\newcommand{\littleand}{\mathbin{\wedge\kern -8truept \wedge}}
\newcommand{\bigor}{\mathop{\bigvee\kern -8.5truept \bigvee}}
\newcommand{\Bigor}{\mathop{\bigvee\kern -10truept \bigvee}}
\newcommand{\littleor}{\mathbin{\vee\kern -8truept \vee}}
\renewcommand{\le}{\leqslant}

\begin{document}

\title{Contractions of quasi relation algebras  and applications to representability}
\titlerunning{Contractions of quasi relation algebras}
%
\author{Andrew Craig \inst{1,2}\orcidID{0000-0002-4787-3760} \and
Wilmari Morton \inst{1}\orcidID{0000-0003-0948-856X} \and
Claudette Robinson\inst{1}\orcidID{0000-0001-7789-4880}}
\authorrunning{A. Craig, W. Morton, C. Robinson}
%
\institute{Department of Mathematics and Applied Mathematics, University of Johannesburg, Auckland Park 2006, South Africa \\
\email{\{acraig,wmorton,claudetter\}@uj.ac.za}
\and
National Institute for Theoretical and Computational Sciences (NITheCS), Johannesburg, South Africa
}

\maketitle

\begin{abstract}
Quasi relation algebras (qRAs) were first described by Galatos and Jipsen in 2013. They are generalisations of relation algebras and can also be viewed as certain residuated lattice expansions. We identify positive symmetric idempotent elements in qRAs and show that they can be used to construct new qRAs, so-called contractions of the original algebra. 
We then show 
that the contraction of a distributive qRA will be  representable when the original algebra is representable. Further, we identify a class of distributive qRAs that are not finitely representable.  
\end{abstract}


\section{Introduction}\label{sec:intro}

One of the central problems in the study of relation algebras is determining which abstract relation algebras can be represented as algebras of binary relations~\cite{JT48}. A common method is to construct new algebras from representable ones, and to show that these new algebras are also then representable. Examples of this include Comer's $A[B]$ construction for relation algebras~\cite{Comer83}, the canonical extension of a relation algebra (check JT1 or JT2), and the  $\mathbf{K}[\mathbf{L}]$ construction for distributive quasi relation algebras~\cite{CMR25}.

Equivalence elements in relation algebras were identified already in the 1950s by J\'{o}nsson and Tarski~\cite{JT52}. For a relation algebra $\mathbf{A}=\langle A, \wedge, \vee, \neg, \top, \bot, \cdot, ^{\smallsmile},1\rangle$, an element $e\in A$ is an equivalence element if $e\cdot e=e = e^{\smallsmile}$. Using an equivalence element 
$e \in A$, 
there are two methods for constructing a new relation algebra 
from~$\mathbf{A}$ (cf.~\cite{J82}). 
One of these methods constructs an algebra $e\mathbf{A}e$ with   $eAe=\{\,e\cdot a\cdot e \mid a \in A\,\}$ as the underlying set. The identity element of the monoid operation on $eAe$ is $e$, the top element is $e \cdot \top\cdot e $ and for $x \in eAe$ its complement is $\neg x \wedge (e \cdot \top \cdot e)$. Givant and Andr\'{e}ka~\cite[Chapter 7]{GA-Simp} refer to this as a \emph{quotient} construction (modulo the equivalence element $e$). If $\mathbf{A}$ is a representable relation algebra, then $e\mathbf{A}e$ is representable~\cite{McK66}. 

In this paper we adapt the use of equivalence elements in relation algebras to the setting of quasi relation algebras (qRAs). These algebras were introduced by Galatos and Jipsen~\cite{GJ13}. Broadly, qRAs can be described as   residuated lattices with three order-reversing negation-like unary operations satisfying certain conditions. Our adaptation of equivalence elements are  \emph{positive symmetric idempotents}, where a  symmetric element is one with the  same image under all three unary operations. 
Our restriction to positive elements is inspired by the use of positive idempotents for a construction for GBI-algebras by Galatos and Jipsen~\cite{GJ20-AU}. 
For a qRA $\mathbf{A}$ and a positive symmetric idempotent element $p \in A$, we define the algebra $p\mathbf{A}p$ whose underlying set is $\{\,p\cdot a\cdot p \mid a \in A\,\}$ and show that it is a qRA (Theorem~\ref{Theorem:RelativizationDqRA}). Following the terminology of \cite[Section 7.2]{GA-Simp}, we call $p\mathbf{A}p$ a \emph{contraction} of $\mathbf{A}$. The contraction of a qRA need not be a subalgebra nor a homomorphic image.

In the spirit of relation algebras, two of the current authors gave a definition for a distributive quasi relation algebra to be \emph{representable}~\cite{RDqRA25}. In essence this requires that the algebra is representable as an algebra of binary relations, where 
the monoid operation  
is modelled  by relational composition. The algebra of binary relations is constructed via a partially ordered set $(X,\leqslant)$ equipped with an equivalence relation $E$, an order automorphism $\alpha$ and a dual order automorphism $\beta$ (see Section~\ref{sec:RDqRA}).

We show that when a distributive quasi relation algebra (DqRA) is representable, then any contraction will also be representable (Theorem~\ref{thm: pAp_representable}). This result helps to develop a catalogue of representable DqRAs. In addition, the method of contractions provides insights into which DqRAs are not finitely representable. Studies of generalisations of relation algebras and their (finite) representability has been studied in many contexts and has applications in programming semantics (see \cite{HS21}, in particular the introduction).

In Section~\ref{sec:prelim} we recall basic definitions of the algebras under consideration and we outline what it means for a DqRA to be representable. We  describe the construction of $p\mathbf{A}p$ in Section~\ref{sec:contractions} before showing in the next section that a contraction of a representable DqRA is representable.  The basic method is to produce a quotient structure of the poset used to represent $\mathbf{A}$ and to show that it can be used to represent $p\mathbf{A}p$. In the final section 
we give an application which is a condition under which a DqRA cannot be represented using a finite poset. 

\section{Preliminaries}\label{sec:prelim}

\subsection{Quasi relation algebras}
Here we recall the basic definitions of the algebras under consideration. For more details the reader is referred to 
\cite{GJKO,GJ13,RDqRA25}.

An algebra $\mathbf{A}=\langle A,\wedge,\vee, \cdot,\backslash,/,1\rangle$ is 
a {\em residuated lattice} (RL) if $\langle A, \cdot,1\rangle$
is a monoid, $\langle A, \wedge, \vee\rangle$ is a lattice
and the monoid operation $\cdot$ is residuated with
residuals $\backslash$ and $/$ such that for all $a,b,c\in A$,
\[a\cdot b\leqslant c\quad\iff\quad 
a\leqslant c/b\quad\iff\quad
b\leqslant a\backslash c.\]
If the underlying lattice is distributive, then $\mathbf{A}$
is said to be a {\em distributive residuated lattice}.
If $a \cdot b = b \cdot a$ for all $a, b \in A$, then
$\mathbf{A}$ is said to be {\em commutative}.

A {\em Full Lambek} (\emph{FL})-\emph{algebra} is a residuated lattice
expanded with a constant $0$, i.e.,
$\mathbf{A}=\langle A,\wedge,\vee, \cdot,\backslash,/,1,0\rangle$. The constant $0$ need not satisfy any additional properties. Linear negations, 
${\sim}:A\to A$ and $-:A\to A$, are defined on 
an FL-algebra in terms of the residuals and $0$ as follows: ${\sim} a=a\backslash 0$ and $ {-}a=0/a$. 

By residuation, 
${\sim} (a\vee b)={\sim}a \wedge {\sim}b$ and ${-}(a\vee b)=
{-a}\wedge{-}b$ for all $a,b\in A$ and we have that
${\sim}1 = 1\backslash 0 = 0= 0/1=-1$.

An {\em involutive Full Lambek} (\emph{InFL})-\emph{algebra} $\mathbf{A}$
is an FL-algebra that satisfies the condition
\[\textsf{(In)}:\quad   {\sim}{-}a=a={-}{\sim}a,\text{ for all }a\in A.\] 
If $\mathbf{A}$ is an InFL-algebra, then $-$ and $\sim$ 
are dual lattice isomorphisms since they are both order-preserving.  Furthermore, for all $a, b\in A$ it holds that
\begin{equation}\label{Eqn:InFLprop} a\leqslant b\quad \iff \quad  a\,\cdot({\sim} b)\leqslant -1 \quad \iff \quad (-b)\cdot a\leqslant-1. \tag{$\ast$} \end{equation}
Define a binary operation $+:A\times A\to A$ by $a+b={\sim}(-a\cdot -b)$ for all $a,b\in A$.
Then $+$ is the dual of $\cdot$ and $a+b=-({\sim} b\cdot {\sim} a)$ holds for all $a,b\in A$ where $\mathbf{A}$ is
an InFL-algebra.

In~\cite[Lemma 2.2]{GJ13} it is shown that an InFL-algebra is term-equivalent to an algebra $\mathbf{A}=\langle A, \wedge,\vee,\cdot,{\sim},{-},1\rangle$ such that $\langle A, \wedge,\vee\rangle$ is a lattice, $\langle A, \cdot, 1\rangle$ is a monoid, and for all $a, b, c \in A$, we have
\begin{equation}\label{Equation:EquivalentSignature}
a\cdot b\leqslant c\quad\iff\quad 
a\leqslant -\left(b\cdot {\sim}c\right)\quad\iff\quad
b\leqslant {\sim}\left(-c\cdot a\right). 
\tag{$\ast\ast$}
\end{equation}
The residuals can be expressed in terms of $\cdot$ and the linear negations as follows:
\[c/b = -\left(b\cdot {\sim}c\right)
    \quad\text{and}\quad
    a\backslash c = {\sim}\left(-c\cdot a\right).
\]
We will  use the signature $\mathbf A = \langle A,\wedge, \vee, \cdot, {\sim}, {-}, 1\rangle$ for an InFL-algebra. 
We will usually  omit the monoid operation $\cdot$ 
in the 
sequel
and simply write $ab$ to denote $a\cdot b$
for the sake of streamlining notation.

An {\em InFL$'$-algebra} $\mathbf{A}$ is an InFL-algebra expanded 
with a unary operation ${\neg}:A\to A$ ($\neg$ is used instead of $'$ for the sake of readibility) such that $\neg\neg a=a$ for all $a\in A$, i.e., 
$\mathbf{A}=\langle A,\wedge, \vee, \cdot, {\sim}, {-}, \neg, 1\rangle$.
A {\em DmInFL$'$-algebra} $\mathbf{A}$ is a InFL$'$-algebra 
that additionally satisfies the De Morgan law
\[\textsf{(Dm)}:\quad\neg (a\vee b)=\neg a\wedge \neg b,\text{ for all }a,b,\in A.\]
A {\em quasi relation algebra} (qRA) is a DmInFL$'$-algebra 
$\mathbf{A}=\langle A,\wedge,\vee,\cdot,{\sim},{-},{\neg},1\rangle$
that additionally satisfies the 
De Morgan product rule, i.e., for all $a,b\in A$,
\[\textsf{(Dp)}:\quad\neg 
(ab)
=\neg a+\neg b.\]
In a qRA $\mathbf{A}$ it holds that 
$\neg 1 = - 1= {\sim}1 = 0$ and the de Morgan involution 
\textsf{(Di)}~$\neg({\sim} a)=-(\neg a)$ holds for all $a \in A$~\cite[Lemma 1]{CJR24}. 

If the underlying lattice $\langle A, \wedge,\vee\rangle$
of a qRA $\mathbf{A}$ is distributive, then $\mathbf{A}=\langle A,\wedge,\vee,\cdot,{\sim},{-},{\neg},1\rangle$
is said to be a {\em distributive quasi relation algebra} (DqRA). An element $a \in A$ is called {\em positive} if $1\leqslant a$, {\em symmetric} if ${\sim}a={-}a=\neg a$ and \emph{idempotent} if 
$a^2 = a$.

\begin{example}\label{Example:DqRAExample}
The six element algebra  in Figure~\ref{fig:DqRAExample} is an example of a 
(distributive) qRA.
It has the identifier $D^6_{3,5,2}$. (See~\cite{CJR-DqRA} for a complete list of DInFL-algebras and DqRAs up to size 8.) In any 
qRA
the monoid identity and
the top element (if it has one) are trivial examples
of positive symmetric idempotents. In $D^6_{3,5,2}$
the elements $a$ and $b$ are non-trivial examples
of positive symmetric idempotents.
\end{example}

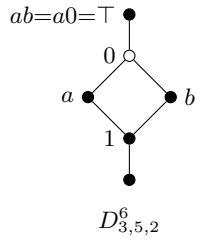
\begin{figure}[h]
    \centering
    \begin{tikzpicture}
           \begin{scope}[xshift=0cm,scale=0.55]
    \node[draw,circle,inner sep=1.5pt,fill] (bot) at (0,0) {};
    \node[draw,circle,inner sep=1.5pt,fill] (1) at (0,1) {};
    \node[draw,circle,inner sep=1.5pt,fill] (a) at (-1,2) {};
    \node[draw,circle,inner sep=1.5pt, fill] (b) at (1,2) {};
     \node[draw,circle,inner sep=1.5pt] (0) at (0,3) {};
    \node[draw,circle,inner sep=1.5pt,fill] (top) at (0,4) {};
    \path [-] (bot) edge node {} (1);
    \path [-] (1) edge node {} (a);
    \path [-] (1) edge node {} (b);
    \path [-] (a) edge node {} (0);
    \path [-] (b) edge node {} (0);
    \path [-] (0) edge node {} (top);
    \node[label,anchor=east,xshift=1pt] at (top) {$ab{=}a0{=}\top$};
    \node[label,anchor=east,xshift=1pt] at (0) {$0$};
    \node[label,anchor=east,xshift=1pt] at (a) {$a$};
    \node[label,anchor=west,xshift=-1pt] at (b) {$b$};
    \node[label,anchor=east,xshift=1pt] at (1) {$1$};
    \node[] at (0,-1) {$D^6_{3,5,2}$};
    \end{scope} 
    \end{tikzpicture}    
        \caption{$D^6_{3,5,2}$ is a DqRA with
        positive symmetric idempotents $\top,1,a$ and $b$.}
    \label{fig:DqRAExample}
\end{figure}

\subsection{Representable DqRAs}\label{sec:RDqRA}
Below we recall the  construction of DqRAs from partially ordered equivalence relations  and hence also the definition of \emph{representability} of a DqRA~\cite{RDqRA25}. 
The lattice structure in the construction is the same as that used in~\cite{GJ20-AU,GJ20-ramics,JS23}.

Let $X$ be a set, $E$ an equivalence relation on $X$ and $R\subseteq X^2$. The converse relation of $R$ is 
$R^\smile = \left\{\left(x, y\right) \mid \left(y, x\right) \in R\right\}$, and its complement (relative to $E$) is
$R^c=\{\, (x,y) \in E \mid (x,y) \notin R\,\}$.
The identity relation is denoted by $\mathrm{id}_X=\{\,(x,x)\mid x \in X\,\}$.
For $R\subseteq E$ we have $(R^{\smile})^{\smile}=R$, $(R^{\smile})^c=(R^c)^{\smile}$, and $\mathrm{id}_X\mathbin{;} R = R\mathbin{;} \mathrm{id}_X = R$.
The composition of two binary relations $S$ and $R$
is given by 
$R \mathbin{;} S = \left\{\left(x, y\right) \mid \left(\exists z \in X\right)\left(\left(x, z\right) \in R \textnormal{ and } \left(z, y\right) \in S\right)\right\}$.
For $R,S, T\subseteq E$ we have $\left(R\, ; S\right)\mathbin{;} T = R\mathbin{;} \left(S\mathbin{;} T\right)$ and 
$\left(R\mathbin{;} S\right)^\smile = S^\smile\mathbin{;} R^\smile$. 

If $\gamma: X \to X$ is a function, then $\gamma$ will be
used to denote either the function or the binary relation 
$\{(x,\gamma(x)) \mid x\in X\}$ that is its graph. 
The following lemma will be used frequently.
Note that it applies when $\gamma$ is a bijective function $\gamma : X \to X$.
\begin{lemma}{\normalfont \cite[Lemma 3.4]{RDqRA25}}
\label{Lemma:ComplementComposition}
Let $E$ be an equivalence relation on a set $X$, and let $R, S, \gamma \subseteq E$. If 
$\gamma$
satisfies 
$\gamma^{\smile}\mathbin{;} \gamma = \mathrm{id}_X$ and $\gamma \mathbin{;} \gamma^{\smile}=\mathrm{id}_X$
then the following hold:
\begin{enumerate}[label=\textup{(\roman*)}]
\item $\left(\gamma\mathbin{;} R\right)^c = \gamma \mathbin{;} R^c$
\item $\left(R\mathbin{;} \gamma\right)^c = R^c \mathbin{;} \gamma$
\end{enumerate}	
\end{lemma}

The following lemma will also be used later.

\begin{lemma}\label{lem:alpha_E_alpha}
Let $X$ be a set,  $E$ an equivalence relation on $X$, and $\gamma: X\to X$ such that $\gamma \subseteq E$. Then  $(x,y) \in E$ iff $(\gamma(x), \gamma(y)) \in E$. 
\end{lemma}

\begin{proof}
Let $(x, y) \in E$. Now since $(x, \gamma(x)) \in \gamma$ and $\gamma \subseteq E$, we have $(x, \gamma(x)) \in E$. Hence, since $E$ is symmetric, $(\gamma(x), x)\in E$, and so, by the transitivity of $E$, $(\gamma(x), y)\in E$. We also have $(y, \gamma(y)) \in \gamma \subseteq E$, so another application of the transitivity of $E$ yields $(\gamma(x), \gamma(y)) \in E$. 

Conversely, let $(\gamma(x), \gamma(y)) \in E$. Now 
$(x,\gamma(x))
\in \gamma\subseteq E$, and so, by the transitivity of $E$, we get $(x, \gamma(y)) \in E$. Moreover, $(y, \gamma(y)) \in \gamma\subseteq E$, and so, by the symmetry and transitivity of $E$, we obtain $(x, y) \in E$. 
\qed
\end{proof}

Let $\mathbf X = \left(X, \leqslant\right)$ be a poset
and $E$ an equivalence relation on $X$ such that ${\leqslant} \subseteq E$. Define a binary relation $\preccurlyeq$
on $E$ as follows:  $(u, v) \preccurlyeq (x, y)  \textnormal{ iff } x \leqslant u \textnormal{ and } v\leqslant y$ for all $(u, v), (x, y) \in E$. Then $\mathbf E = \left(E, \preccurlyeq\right)$ is a poset and hence the set of up-sets of $\mathbf E$ ordered by inclusion, denoted by $\textsf{Up}\left(\mathbf E\right)$,
is a distributive lattice.
Let $\mathsf{Down}(\mathbf{E})$ denote the set of downsets of the poset 
$(E,\preccurlyeq)$. Both $\UpE$ and $\mathsf{Down}(\mathbf{E})$ are closed under composition.
Further, $R \in \UpE$ iff $R^c \in \mathsf{Down}(\mathbf{E})$ iff $R^{\smile} \in \mathsf{Down}(\mathbf{E})$.

Note that
${\leqslant}\in \mathsf{Up}(\mathbf{E})$ and it is the identity with respect to the composition,~$\mathbin{;}$. 
Furthermore, the composition, $\mathbin{;}$, is residuated with its residuals defined by
$R \backslash_{\UpE} S = (R^{\smile}\mathbin{;}S^c)^c$ and 
$R/_{\UpE}S=(R^c\mathbin{;}S^{\smile})^c$.
Therefore, the algebra $\langle \mathsf{Up}(\mathbf{E}), \cap, \cup, \mathbin{;},\backslash_{\UpE},/_{\UpE}, \leqslant \rangle
$ is a distributive residuated lattice.

Next let $\alpha : X \to X$ be an order automorphism and  $\beta : X \to X$ a self-inverse dual order automorphism.
If $R, S  \in \mathsf{Up}\left(\mathbf E\right)$, then 
$\alpha\mathbin{;} R\in\UpE$ and  $R\mathbin{;} \alpha\in\UpE$.
If $R  \in \mathsf{Down}\left(\mathbf E\right)$, then $
\beta\mathbin{;}R\mathbin{;}\beta \in \UpE$.
See~\cite[Lemma 3.5]{RDqRA25} for more details.
For all $R \in \mathsf{Up}\left(\mathbf{E}\right)$,
the linear negations can be defined by:
\[{\sim} R = R^{c\smile}\mathbin{;} \alpha\text{ and } {-}R = \alpha \mathbin{;} R^{c\smile}.\]
Lastly, $\neg R$ is defined using the self-inverse dual order automorphism $\beta$, as stated below.
\begin{theorem}{\normalfont \cite[Theorem 3.15]{RDqRA25}}\label{Theorem:Dq(E)}
Let $\mathbf{X}=\left(X,\leqslant\right)$ be a poset and $E$ an equivalence relation on $X$ such that ${\leqslant} \subseteq E$.  Let $\alpha: X \to X$ be an order automorphism of $\mathbf X$ and $\beta: X \to X$ a self-inverse dual order automorphism of $\mathbf X$ such that $\alpha, \beta \subseteq E$ and $\beta = \alpha \mathbin{;} \beta\mathbin{;} \alpha$. 
Set $1={\leqslant}$ and 
for 
$R \in \mathsf{Up}(\mathbf E)$, define
${\sim} R = R^{c\smile}\mathbin{;} \alpha$, $-R = \alpha \mathbin{;} R^{c\smile}$, and 
$\neg R =  \alpha\mathbin{;} \beta \mathbin{;} R^c \mathbin{;} \beta$.
Then the algebra $\mathbf{Dq}(\mathbf E) = \left\langle \mathsf{Up}\left(\mathbf{E}\right),\cap, \cup, \mathbin{;}, {\sim}, {-}, {\neg}, 1 \right\rangle$ 
is a distributive quasi relation algebra. If $\alpha$ is the identity, then $\mathbf{Dq}(\mathbf E)$ is a cyclic distributive quasi relation algebra.	
\end{theorem}	

Algebras of the form described by Theorem~\ref{Theorem:Dq(E)} are called \emph{equivalence distributive quasi relation algebras} and the class of such algebras is denoted $\mathsf{EDqRA}$. The algebra $\mathbf{Dq}(\mathbf{E})$ is
said to be a \emph{full distributive quasi relation algebra}
if $E=X^2$ and the class of such algebras denoted is by $\mathsf{FDqRA}$. Analogous to the case for relation algebras (cf.~\cite[Chapter 3]{Mad06}), 
it was shown~\cite[Theorem 4.4]{RDqRA25} that 
$\mathbb{IP}(\mathsf{FDqRA})=\mathbb{I}(\mathsf{EDqRA})$. 
\begin{definition}
{\normalfont \cite[Definition 4.5]{RDqRA25}}\label{Definition:RDqRA}
A DqRA $\mathbf{A} = \left\langle A, \wedge, \vee, \cdot,  {\sim}, {-}, {\neg}, 1 \right\rangle$ 
is \emph{representable} if 
$\mathbf{A} \in \mathbb{ISP}\left(\mathsf{FDqRA}\right)$
or, equivalently, $\mathbf{A} \in \mathbb{IS}\left(\mathsf{EDqRA}\right)$. 
\end{definition}

A DqRA $\mathbf{A}$ is \emph{finitely} representable if the poset $( X ,\leqslant)$ used in the representation of $\mathbf{A}$ is finite. 

\begin{example}\label{Example:DqRARepresentation}
Let $X=\{w,x,y,z\}$ be the four element
antichain depicted in Figure~\ref{fig:DqRARepresentation} such
that $E=X^2$, $\alpha: X\to X$ is given by $\alpha=\{(w,x), (x,w), (y,z), (z,y)\}$ and 
$\beta:X\to X$ is given by 
$\beta=\{(w,y),(y,w),(x,z),(z,x)\}$.
Also let
\begin{align*}
R_a&=\{(w,w),(x,x),(y,y),(z,z),(w,y),(y,w),(x,z),(z,x)\} \quad \text{and}\\
R_b&=\{(w,w),(x,x),(y,y),(z,z),(w,z),(z,w),(x,y),(y,x)\}.
\end{align*}
Then $D^6_{3,5,2}$ depicted in Figure~\ref{fig:DqRAExample} is finitely
representable since it is 
isomorphic 
to
the subalgebra of $\mathbf{Dq}(\mathbf E)$ with elements  $\{\varnothing, \leqslant, R_a, R_b, \alpha;\leqslant^{c\smile},X^2\}$.
\end{example}
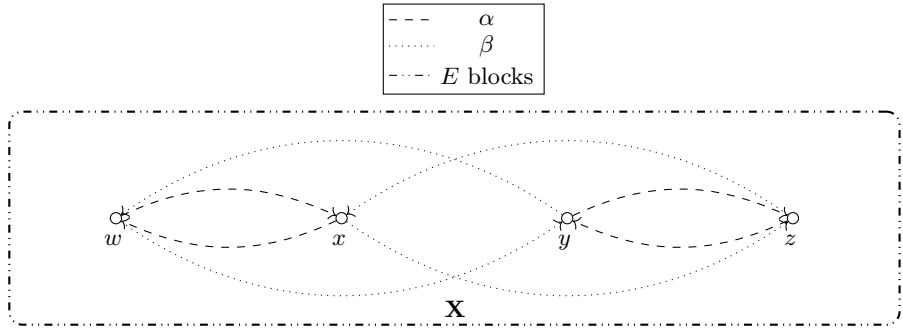
\begin{figure}[h]
    \centering
    \begin{tikzpicture}[scale=1.5,pics/sample/.style={code={\draw[#1] (0,0) --(0.6,0) ;}},
Dotted/.style={
dash pattern=on 0.1\pgflinewidth off #1\pgflinewidth,line cap=round,
shorten >=#1\pgflinewidth/2,shorten <=#1\pgflinewidth/2},
Dotted/.default=3]
\begin{scope}[xshift=-5cm, box/.style = {draw,dashdotdotted,inner sep=38pt,rounded corners=5pt,thick}]
\node[draw,circle,inner sep=1.5pt] (w) at (-2,0.7) {};
\node[draw,circle,inner sep=1.5pt] (x) at (0,0.7) {};
\node[draw,circle,inner sep=1.5pt] (y) at (2,0.7) {};
\node[draw,circle,inner sep=1.5pt] (z) at (4,0.7) {};
\path (w) edge [->, bend left=25, dashed] node {} (x);
\path (x) edge [->, bend left=25, dashed] node {} (w);
\path (y) edge [->, bend left=25, dashed] node {} (z);
\path (z) edge [->, bend left=25, dashed] node {} (y);
\path (y) edge [->, bend right=35,dotted] node{} (w);
\path (w) edge [->, bend right=35,dotted] node{} (y);
\path (x) edge [->, bend right=35,dotted] node{} (z);
\path (z) edge [->, bend right=35,dotted] node{} (x);
\node[label,anchor=north,xshift=-1pt] at (w) {$w$};
\node[label,anchor=north,xshift=-1pt] at (x) {$x$};
\node[label,anchor=north,xshift=-1pt] at (y) {$y$};
\node[label,anchor=north,xshift=-1pt] at (z) {$z$};
\path (1.8,2.2) 
 node[matrix,anchor=east,draw,nodes={anchor=center},inner sep=2pt]  {
  \pic{sample=dashed}; & \node{$\alpha$}; \\
  \pic{sample=dotted}; & \node{$\beta$}; \\
  \pic{sample=dashdotdotted}; & \node{$E$ blocks}; \\
 };
 \node[box,fit=(w)(x)(y)(z)] {};
\node[label,anchor=north,xshift=43pt,yshift=-25pt] at (x) {$\mathbf{X}$};
\end{scope}
\end{tikzpicture}
    \caption{The poset used to represent $D^6_{3,5,2}$.}
    \label{fig:DqRARepresentation}
\end{figure}

\section{Contractions of quasi relation algebras}\label{sec:contractions}


In this section we present the construction of contractions of a qRA.  
Consider a qRA  $\mathbf{A}=\langle A, \wedge, \vee, \cdot, \sim, -,\neg, 1 \rangle$
and   a positive symmetric idempotent element $p$ of $\mathbf{A}$.
Let   $pAp := \{\, pap \mid a \in A\}$. We show that $pAp$ is closed under all the operations of $\mathbf{A}$ except the monoid identity, which we redefine. 

\begin{lemma}\label{lem:TFAE-pAp}
Let $\A$ be a qRA and $p$ an idempotent element of $\A$. The following are equivalent for any $b \in A$:
\begin{enumerate}[label=\textup{(\roman*)}]
\item $b \in pAp$
\item $pbp=b$
\item $pb=b$ and $bp=b$.
\end{enumerate}
\end{lemma}
\begin{proof}
Assume (i). Then $b=pap$ for some $a \in A$. Since $p$ is idempotent we get $pbp=p(pap)p=pap=b$, showing (ii). If we assume (ii) then $pb=p(pbp)=pbp=b$ and $bp=(pbp)p=pbp=b$. Lastly, assume (iii). We have $b=pb=p(bp)=pbp$ so clearly $b \in pAp$. \qed 
\end{proof}

\newpage 
\begin{lemma}\label{lem:pAp-sublat}
Let $\A$ be a qRA and $p$ a positive symmetric idempotent.  Then $\langle pAp,\wedge,\vee\rangle$ forms a sublattice of $\langle A, \wedge, \vee\rangle$ and $\langle pAp,\cdot\rangle $ is a subsemigroup of $\langle A,\cdot\rangle$. 
\end{lemma}
\begin{proof}
Let $c,d \in pAp$. Then $c=pcp$ and $d=pdp$ by Lemma~\ref{lem:TFAE-pAp}. The join-preservation of $\cdot$ gives us: $c \vee d = pcp \vee pdp = (pc \vee pd)p = p(c \vee d)p$, so $pAp$ is closed under joins. 

By Lemma~\ref{lem:TFAE-pAp}(iii) we have $pc=c$ and $pd=d$. Since $p$ is positive,  
$c \wedge d \leqslant p(c \wedge d)$. From $c \wedge d \leqslant c,d$ we get $p(c \wedge d) \leqslant pc,pd $ and hence $p(c \wedge d) \leqslant pc \wedge pd=c \wedge d$. This gives us $c \wedge d  = p(c \wedge d)$. One can similarly use the fact that $c=cp$ and $d=dp$ to prove that $c \wedge d = (c \wedge d)p$, giving $c \wedge d \in pAp$ by Lemma~\ref{lem:TFAE-pAp}. 

For the monoid operation, we get 
$cd=(pcp)(pdp)= p(cp^2d)p$,
and thus $cd \in pAp$, with the associativity of $\cdot$ simply inherited. \qed
\end{proof}

We emphasize that if $\langle A, \wedge,\vee\rangle$ is distributive, then so is $\langle pAp,\wedge,\vee \rangle$.   
For $p\neq 1$ the set $pAp$ will not form a submonoid of $\langle A, \cdot, 1\rangle$ since $1 \notin pAp$. Now we prove that $pAp$ is closed under the unary operations. 

\begin{lemma}\label{lem:pAp-clos-unary}
Let $\A$ be a qRA, $p$ a positive symmetric idempotent, and $b \in pAp$. Then
\begin{enumerate}[label=\textup{(\roman*)}]
\item ${\sim}b \in pAp$;
\item $-b \in pAp$;
\item $\neg b \in pAp$. 
\end{enumerate}
\end{lemma}
\begin{proof}
To prove item (i), we will show $\tw b = p \cdot \tw b \cdot p$. 
Since $p$ is positive we obtain $\tw b \leqslant p(\tw b)p$. By Lemma~\ref{lem:TFAE-pAp} we have $b=pbp$ and hence $\tw b \leqslant \tw (pbp)$. By~\textsf{(In)}, 
 $p=-\tw p$, so we get 
$\tw b \leqslant \tw (-(\tw p)\cdot (bp))$, and then applying  
(\ref{Equation:EquivalentSignature}) gives us 
$(bp)({\sim}b) \leqslant {\sim}p$. Since $p$ is symmetric, $(bp)({\sim}b) \leqslant -p$, and since $p$ is positive, $b(p\cdot \tw b)p\leqslant -p\cdot p$. By (\ref{Eqn:InFLprop}) we get $-p\cdot p\leqslant -1$, as $p\leqslant p$. 
Hence 
$b(p\cdot \tw b \cdot p) \leqslant -1$ and again by
(\ref{Equation:EquivalentSignature})
we get  $p\cdot \tw b\cdot p \leqslant \tw ({--}1\cdot b)=\tw b$. 

Similarly, one can use
(\ref{Equation:EquivalentSignature})
and (\ref{Eqn:InFLprop})
to show that $-b=p(-b)p$. 

For (iii), we use 
(\textsf{Dp})
to get $\neg b = \neg (pbp)=\neg (pb) + \neg p = -({\sim}\neg p \cdot {\sim}\neg (pb))$. Further, 
(\textsf{Dp})
gives $\neg(pb)= \neg p+\neg b= -(\tw \neg b \cdot \tw \neg p)$. Hence $\neg b = \neg (pbp)=-(\tw\neg p \cdot \tw(-(\tw \neg b \cdot \tw \neg p)))=-(\tw\neg p \cdot \tw\neg b \cdot \tw\neg p)=-(p(\tw\neg b)p)$, by \textsf{(In)} since $p$ is symmetric.  By~(ii), this is an element of $pAp$. \qed    
\end{proof}

\begin{theorem}\label{Theorem:RelativizationDqRA}
Let $\mathbf{A} = \langle A,\wedge, \vee, \cdot, {\sim},-,\neg, 1\rangle $ be a quasi relation algebra and $p$ a positive symmetric idempotent. Then $p\mathbf{A}p=\langle pAp, \wedge,\vee, \cdot, \tw, -,\neg,p \rangle $ is a quasi relation algebra.   
\end{theorem}
\begin{proof} 
After Lemmas~\ref{lem:pAp-sublat} and~\ref{lem:pAp-clos-unary} we only need to check that $p$ is the identity of $pAp$ with respect to the monoid operation. By Lemma~\ref{lem:TFAE-pAp}(ii), for $b \in pAp$ we get $bp=(pbp)p=pbp=b=p^2bp=pb$. 

%
All other qRA equations 
(except for the monoid identity equation)
involve operations that remain unchanged on elements of $pAp$, so the equations are satisfied by all elements of $pAp$. \qed
\end{proof}

\begin{example}\label{Example:ConstructionExample}
Recall from Example~\ref{Example:DqRAExample}
that $D^6_{3,5,2}$ depicted in Figure~\ref{fig:DqRAExample} has four
positive symmetric idempotents: $1$, $\top$, $a$
and $b$. From Theorem~\ref{Theorem:RelativizationDqRA} it
follows that each $p\mathbf{A}p$, depicted
in Figure~\ref{fig:ConstructionExample}, 
is a qRA for $\mathbf{A}=D^6_{3,5,2}$ and $p\in\{1,\top,a,b\}$.  Note that for $p=1$ it
will always be the case that 
$p\mathbf{A}p =\mathbf{A}$.
\end{example}

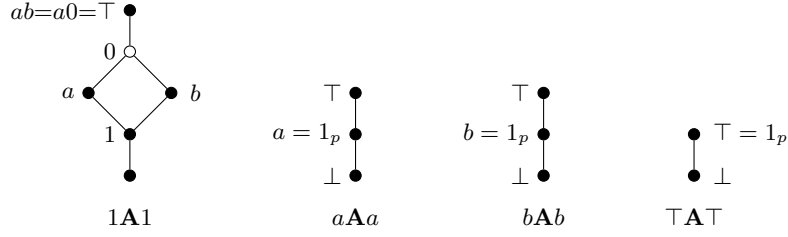
\begin{figure}
    \centering
        \begin{tikzpicture}
           \begin{scope}[xshift=-5cm,scale=0.55]
    \node[draw,circle,inner sep=1.5pt,fill] (bot) at (0,0) {};
    \node[draw,circle,inner sep=1.5pt,fill] (1) at (0,1) {};
    \node[draw,circle,inner sep=1.5pt,fill] (a) at (-1,2) {};
    \node[draw,circle,inner sep=1.5pt, fill] (b) at (1,2) {};
     \node[draw,circle,inner sep=1.5pt] (0) at (0,3) {};
    \node[draw,circle,inner sep=1.5pt,fill] (top) at (0,4) {};
    \path [-] (bot) edge node {} (1);
    \path [-] (1) edge node {} (a);
    \path [-] (1) edge node {} (b);
    \path [-] (a) edge node {} (0);
    \path [-] (b) edge node {} (0);
    \path [-] (0) edge node {} (top);
    \node[label,anchor=east,xshift=1pt] at (top) {$ab{=}a0{=}\top$};
    \node[label,anchor=east,xshift=1pt] at (0) {$0$};
    \node[label,anchor=east,xshift=1pt] at (a) {$a$};
    \node[label,anchor=west,xshift=1pt] at (b) {$b$};
    \node[label,anchor=east,xshift=1pt] at (1) {$1$};
    \node[] at (0,-1) {$1\mathbf{A}1$};
    \end{scope} 
    \begin{scope}[xshift=-2cm,scale=0.55]
    \node[draw,circle,inner sep=1.5pt,fill] (bot) at (0,0) {};
    \node[draw,circle,inner sep=1.5pt, fill] (a) at (0,1) {};
    \node[draw,circle,inner sep=1.5pt,fill] (top) at (0,2) {};
    \path [-] (bot) edge node {} (a);
    \path [-] (a) edge node {} (top);
    \node[label,anchor=east,xshift=1pt] at (top) {$\top$};
    \node[label,anchor=east,xshift=1pt] at (a) {$a=1_p$};
    \node[label,anchor=east,xshift=1pt] at (bot) {$\bot$};    
    \node[] at (0,-1) {$a\mathbf{A}a$};
    \end{scope}
    \begin{scope}[xshift=0.5cm,scale=0.55]
    \node[draw,circle,inner sep=1.5pt,fill] (bot) at (0,0) {};
    \node[draw,circle,inner sep=1.5pt, fill] (a) at (0,1) {};
    \node[draw,circle,inner sep=1.5pt,fill] (top) at (0,2) {};
    \path [-] (bot) edge node {} (a);
    \path [-] (a) edge node {} (top);
    \node[label,anchor=east,xshift=1pt] at (top) {$\top$};
    \node[label,anchor=east,xshift=1pt] at (a) {$b=1_p$};
    \node[label,anchor=east,xshift=1pt] at (bot) {$\bot$};    
    \node[] at (0,-1) {$b\mathbf{A}b$};
    \end{scope}
    \begin{scope}[xshift=2.5cm,scale=0.55]
    \node[draw,circle,inner sep=1.5pt,fill] (bot) at (0,0) {};
    \node[draw,circle,inner sep=1.5pt,fill] (top) at (0,1) {};
    \path [-] (bot) edge node {} (top);
    \node[label,anchor=west,xshift=1pt] at (top) {$\top=1_p$};
    \node[label,anchor=west,xshift=1pt] at (bot) {$\bot$};    
    \node[] at (0,-1) {$\top \mathbf{A}\top$};
    \end{scope}
    \end{tikzpicture} 
\caption{The 
algebras $p\mathbf{A}p$
    for $\mathbf{A}=D^6_{3,5,2}$ and $p\in\{1,\top,a,b\}$. }
\label{fig:ConstructionExample}
\end{figure}

\section{Representability of contractions of DqRAs}
\label{sec: pAp_representable}

In this section, we prove that if $\mathbf{A}$ is a representable distributive quasi relation algebra and $p$ is a positive symmetric  idempotent of $\mathbf{A}$, then its contraction $p\mathbf{A}p$ is also representable.

Let $\mathbf{A} = \langle A,\wedge, \vee, \cdot, {\sim},{-},{\neg}, 1\rangle$ be a representable DqRA. Then 
there exists a poset $\mathbf X_\mathbf{A} = \left(X_\mathbf{A}, \leqslant_{X_\mathbf{A}}\right)$,
an equivalence relation $E_\mathbf{A}\subseteq X_\mathbf{A}\times X_\mathbf{A}$ with
${\leqslant_{X_\mathbf{A}}}\subseteq E_\mathbf{A}$, an
order automorphism $\alpha_\mathbf{A}:X_\mathbf{A}\to X_\mathbf{A}$
and a self-inverse dual order automorphism $\beta_\mathbf{A}:X_\mathbf{A}\to X_\mathbf{A}$ such that 
$\alpha_\mathbf{A},\beta_\mathbf{A}\subseteq E_\mathbf{A}$ and
$\alpha_\mathbf{A}\mathbin{;} \beta_\mathbf{A} \mathbin{;} \alpha_\mathbf{A} = \beta_\mathbf{A}$.  Moreover, there exists an embedding
$\varphi: \mathbf{A} \hookrightarrow \mathbf{Dq}\left(E_\mathbf{A}, \preccurlyeq_{E_\mathbf{A}}\right)$.

It is easy to check that if $p$ is a positive symmetric idempotent of $\mathbf{A}$, then 
$\varphi(p)$ is a preorder on $X_\mathbf{A}$ such that ${\leqslant_{X_\mathbf{A}}} \subseteq \varphi(p)$. Now define ${\equiv} \subseteq X_\mathbf{A} \times X_\mathbf{A}$ by 
$$x \equiv y \quad \text{ iff } \quad (x,y) \in \varphi(p) \text{ and } (y,x) \in \varphi(p).$$
Then $\equiv$ is an equivalence relation on $X_\mathbf{A}$. For each $x \in X_\mathbf{A}$, we will denote the equivalence class of $x$ with respect to the equivalence relation $\equiv$ by $[x]$. Now let ${X_\mathbf{A}/{\equiv}} = \{[x] \mid x \in X_\mathbf{A}\}$. Define  ${\leqslant_{X_\mathbf{A}/{\equiv}}} \subseteq X_{\mathbf{A}}/{\equiv} \times X_{\mathbf{A}}/{\equiv}$ by
$$[x]\leqslant_{X_\mathbf{A}/\equiv} [y] \quad \text{ iff } \quad (x, y) \in \varphi(p).
$$
It is straightforward to check that this is a well-defined partial order on $X_\mathbf{A}/{\equiv}$.

Next define $E_{p\mathbf{A}p} \subseteq {X_\mathbf{A}/{\equiv}} \times {X_\mathbf{A}/{\equiv}}$ by
\[
([x], [y]) \in E_{p\mathbf{A}p} \quad \textnormal{ iff } \quad (x, y) \in E_\mathbf{A}.
\]
Using the fact that $E_\mathbf{A}$ is an equivalence relation on $X_\mathbf{A}$, we can show that $E_{p\mathbf{A}p}$ is an equivalence relation on $X_\mathbf{A}/{\equiv}$. Moreover, since ${\leqslant_{X_\mathbf{A}}} \subseteq E_\mathbf{A}$, it follows that ${\leqslant_{X_\mathbf{A}/{\equiv}}}\subseteq E_{p\mathbf{A}p}$. 

Now define $\alpha_{p\mathbf{A}p}: {X_\mathbf{A}/{\equiv}} \to {X_\mathbf{A}/{\equiv}}$ by setting, for all $[x]\in X_\mathbf{A}/{\equiv}$, $$\alpha_{p\mathbf{A}p}([x]) = [\alpha_\mathbf{A}(x)].$$
Likewise, define $\beta_{p\mathbf{A}p}: {X_\mathbf{A}/{\equiv}} \to {X_\mathbf{A}/{\equiv}}$ by setting, for all $[x] \in X_\mathbf{A}/{\equiv}$, $$\beta_{p\mathbf{A}p}([x]) = [\beta_\mathbf{A}(x)].$$

\begin{example}\label{Example:Consturction1}
Recall that the DqRA $\mathbf{A}=D^6_{3,5,2}$ from Example~\ref{Example:DqRAExample}, depicted in
    Figure~\ref{fig:DqRAExample}, has four
    positive symmetric idempotent, namely $1,a,b$ and $\top$. Also recall from Example~\ref{Example:DqRARepresentation} that $\mathbf{A}$ is representable using the poset $X$, depicted in Figure~\ref{fig:DqRARepresentation}.
    Using the construction described above,
    we can obtain an equivalence relation ${\equiv}_p$,
    a quotient structure $X/{\equiv}_p$, an equivalence
    relation  $E_{p\mathbf{A}p}$ and the maps 
    $\alpha_{p\mathbf{A}p}$ and $\beta_{p\mathbf{A}p}$
    for each $p\in \{1,a,b,\top\}$.  In this example we have that ${\equiv}_p=\varphi(p)$ and the
    posets obtained are depicted in Figure~\ref{fig:Construction1}.
\end{example}
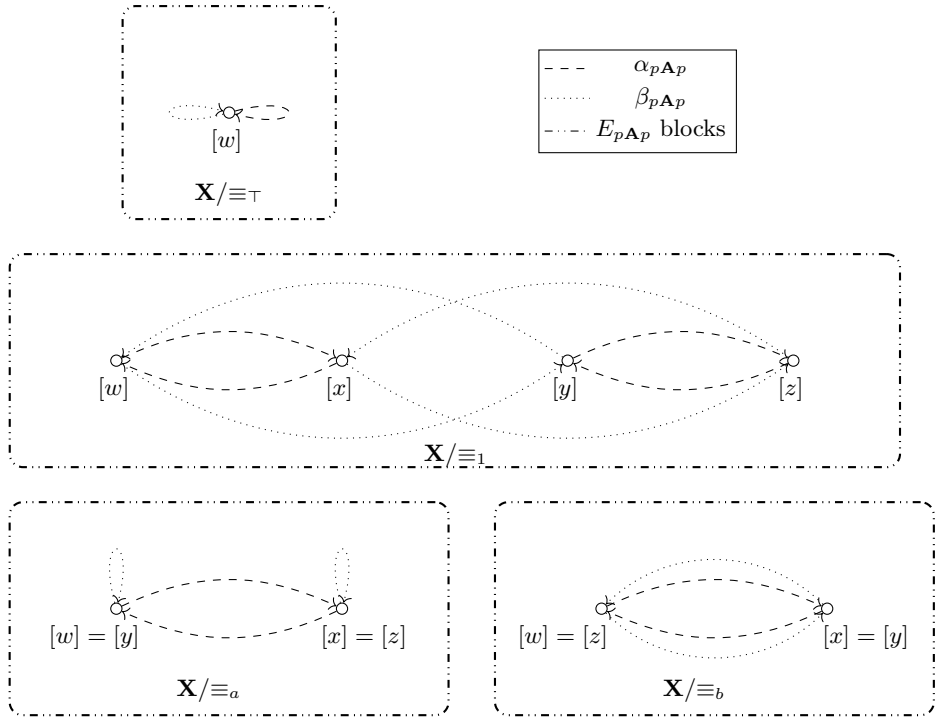
\begin{figure}[h]
    \centering
    \begin{tikzpicture}[scale=1.5,pics/sample/.style={code={\draw[#1] (0,0) --(0.6,0) ;}},
Dotted/.style={
dash pattern=on 0.1\pgflinewidth off #1\pgflinewidth,line cap=round,
shorten >=#1\pgflinewidth/2,shorten <=#1\pgflinewidth/2},
Dotted/.default=3]
\begin{scope}[xshift=-4cm, box/.style = {draw,dashdotdotted,inner sep=38pt,rounded corners=5pt,thick}]
\node[draw,circle,inner sep=1.5pt] (w) at (-2,0.7) {};
\node[draw,circle,inner sep=1.5pt] (x) at (0,0.7) {};
\node[draw,circle,inner sep=1.5pt] (y) at (2,0.7) {};
\node[draw,circle,inner sep=1.5pt] (z) at (4,0.7) {};
\path (w) edge [->, bend left=25, dashed] node {} (x);
\path (x) edge [->, bend left=25, dashed] node {} (w);
\path (y) edge [->, bend left=25, dashed] node {} (z);
\path (z) edge [->, bend left=25, dashed] node {} (y);
\path (w) edge [->, bend right=35,dotted] node{} (y);
\path (y) edge [->, bend right=35,dotted] node{} (w);
\path (x) edge [->, bend right=35,dotted] node{} (z);
\path (z) edge [->, bend right=35,dotted] node{} (x);
\node[label,anchor=north,xshift=-1pt] at (w) {$[w]$};
\node[label,anchor=north,xshift=-1pt] at (x) {$[x]$};
\node[label,anchor=north,xshift=-1pt] at (y) {$[y]$};
\node[label,anchor=north,xshift=-1pt] at (z) {$[z]$};
\path (3.5,3) 
 node[matrix,anchor=east,draw,nodes={anchor=center},inner sep=2pt]  { \pic{sample=dashed}; & \node{$\alpha_{p\mathbf{A}p}$}; \\
  \pic{sample=dotted}; & \node{$\beta_{p\mathbf{A}p}$}; \\
  \pic{sample=dashdotdotted}; & \node{$E_{p\mathbf{A}p}$ blocks}; \\
 };
 \node[box,fit=(w)(x)(y)(z)] {};
\node[label,anchor=north,xshift=43pt,yshift=-25pt] at (x) {$\mathbf{X}/{\equiv}_1$};
\end{scope}
\begin{scope}[xshift=-6cm, yshift=2.2cm, box/.style = {draw,dashdotdotted,inner sep=38pt,rounded corners=5pt,thick}]
\node[draw,circle,inner sep=1.5pt] (e) at (1,0.7) {};
\draw [->,dashed] (e) edge[loop right,looseness=60]node{} (e);
\draw [->,dotted] (e) edge[loop left,looseness=60]node{} (e);
\node[label,anchor=north,xshift=0pt,yshift=-20pt] at (e) {$\mathbf{X}/{\equiv}_\top$};
\node[label,anchor=north,xshift=-1pt] at (e) {$[w]$};
 \node[box,fit=(e)] {};
\end{scope}

\begin{scope}[xshift=-2cm,yshift=-2.2cm, box/.style = {draw,dashdotdotted,inner sep=38pt,rounded corners=5pt,thick}]
\node[draw,circle,inner sep=1.5pt] (a) at (-4,0.7) {};
\node[draw,circle,inner sep=1.5pt] (b) at (-2,0.7) {};
\path (a) edge [->, bend left=25, dashed] node {} (b);
\path (b) edge [->, bend left=25, dashed] node {} (a);
\draw [->,dotted] (a) edge[loop above,looseness=60]node{} (a);
\draw [->,dotted] (b) edge[loop above,looseness=60]node{} (b);
\node[label,anchor=north,xshift=35pt,yshift=-20pt] at (a) {$\mathbf{X}/{\equiv}_a$};
\node[label,anchor=north,xshift=-7pt] at (a) {$[w]=[y]\;$};
\node[label,anchor=north,xshift=7pt] at (b) {$\;[x]=[z]$};
 \node[box,fit=(a)(b)] {};
\end{scope}
\begin{scope}[xshift=-2.7cm, yshift=-2.2cm, box/.style = {draw,dashdotdotted,inner sep=38pt,rounded corners=5pt,thick}]
\node[draw,circle,inner sep=1.5pt] (c) at (1,0.7) {};
\node[draw,circle,inner sep=1.5pt] (d) at (3,0.7) {};
\path (c) edge [->, bend left=25, dashed] node {} (d);
\path (d) edge [->, bend left=25, dashed] node {} (c);
\path (c) edge [->, bend right=45, dotted] node {} (d);
\path (d) edge [->, bend right=45, dotted] node {} (c);
\node[label,anchor=north,xshift=35pt,yshift=-20pt] at (c) {$\mathbf{X}/{\equiv}_b$};
\node[label,anchor=north,xshift=-10pt] at (c) {$[w]=[z]\quad$};
\node[label,anchor=north,xshift=10pt] at (d) {$\quad[x]=[y]$};
 \node[box,fit=(c)(d)] {};
\end{scope}
\end{tikzpicture}
    \caption{The posets $\mathbf X/{\equiv}_p$ for $p\in\{1,a,b,\top\}$}
    \label{fig:Construction1}
\end{figure}

The next two lemmas show that the preorder $\varphi(p)$ is invariant with respect to the maps $\alpha_\mathbf{A}$ and $\beta_\mathbf{A}$. 

\begin{lemma}\label{lem:alpha_order_aut_varphi(p)}
For all $x, y \in X_\mathbf{A}$, we have $(x,y) \in \varphi(p)$ iff $(\alpha_\mathbf{A}(x), \alpha_{\mathbf{A}}(y)) \in \varphi(p)$.
\end{lemma}

\begin{proof}
Assume $(x,y) \in \varphi(p)$. Suppose for the sake of a contradiction that $(\alpha_\mathbf{A}(x), \alpha_{\mathbf{A}}(y)) \notin \varphi(p)$. Since $(x,y) \in \varphi(p)$ and $\varphi(p)\subseteq E_\mathbf{A}$, we have $(x,y) \in E_\mathbf{A}$, and so $(\alpha_\mathbf{A}(x), \alpha_{\mathbf{A}}(y)) \in E_\mathbf{A}$ by Lemma~\ref{lem:alpha_E_alpha}. Hence,  $(\alpha_\mathbf{A}(x), \alpha_{\mathbf{A}}(y)) \in \varphi(p)^c$, and thus  $(\alpha_\mathbf{A}(y), \alpha_{\mathbf{A}}(x)) \in \varphi(p)^{c\smile}$. Now $(y, \alpha_\mathbf{A}(y)) \in \alpha_\mathbf{A}$, so it follows that $(y, \alpha_\mathbf{A}(x)) \in \alpha_\mathbf{A}\mathbin{;} \varphi(p)^{c\smile} = -\varphi(p)$. Since $p$ is symmetric and $\varphi$ preserves the linear negations, we get $-\varphi(p) = \varphi(-p) = \varphi({\sim}p) = {\sim}\varphi(p)$, and consequently $(y, \alpha_\mathbf{A}(x)) \in {\sim}\varphi(p) = \varphi(p)^{c\smile}\mathbin{;} \alpha_\mathbf{A}$. This shows $(y, \alpha^{-1}_\mathbf{A}(\alpha_\mathbf{A}(x))) \in \varphi(p)^{c\smile}$, which means $(x, y) \notin \varphi(p)$, a contradiction. 

For the converse implication, assume $(\alpha_\mathbf{A}(x), \alpha_{\mathbf{A}}(y)) \in \varphi(p)$ and suppose $(x, y) \notin  \varphi(p)$. The first part implies $(x, y) \in E_\mathbf{A}$, and therefore, since we have $(x, y) \notin  \varphi(p)$, it follows that $(y, x) \in \varphi(p)^{c\smile}$. Now $(x, \alpha_\mathbf{A}(x)) \in \alpha_\mathbf{A}$, and so $(y, \alpha_\mathbf{A}(x)) \in \varphi(p)^{c\smile}\mathbin{;}\alpha_\mathbf{A} = {\sim}\varphi(p)$. But $-\varphi(p) = {\sim}\varphi(p)$, so $(y, \alpha_\mathbf{A}(x)) \in -\varphi(p)$. This gives $(\alpha_\mathbf{A}(y), \alpha_{\mathbf{A}}(x)) \in \varphi(p)^{c\smile}$, which means $(\alpha_\mathbf{A}(x), \alpha_{\mathbf{A}}(y)) \notin \varphi(p)$, a contradiction.
\qed
\end{proof}


\begin{lemma}\label{lem:beta_dual_order_aut_varphi(p)}
For all $x, y \in X_\mathbf{A}$, we have $(x,y) \in \varphi(p)$ iff $(\beta_\mathbf{A}(y), \beta_{\mathbf{A}}(x)) \in \varphi(p)$.
\end{lemma}

\begin{proof}
Assume $(x,y) \in \varphi(p)$. Suppose for the sake of a contradiction that $(\beta_\mathbf{A}(y), \beta_{\mathbf{A}}(x)) \notin \varphi(p)$. Since $(x,y) \in \varphi(p)\subseteq E_\mathbf{A}$, we have $(x,y) \in E_\mathbf{A}$, and so $(\beta_\mathbf{A}(y), \beta_{\mathbf{A}}(x)) \in E_\mathbf{A}$ by Lemma~\ref{lem:alpha_E_alpha}. Hence, $(\beta_\mathbf{A}(y), \beta_{\mathbf{A}}(x)) \in \varphi(p)^c$, and so, since $(y, \beta_\mathbf{A}(y)) \in \beta_\mathbf{A}$ and $(\beta_\mathbf{A}(x), x) \in \beta_\mathbf{A}$, we get $(y, x) \in \beta_\mathbf{A}\mathbin{;}\varphi(p)^c\mathbin{;}\beta_\mathbf{A}$. We also know that $(\alpha^{-1}_\mathbf{A}(y), y) \in \alpha_\mathbf{A}$, and therefore $(\alpha^{-1}_\mathbf{A}(y), x) \in \alpha_\mathbf{A}\mathbin{;}\beta_\mathbf{A}\mathbin{;}\varphi(p)^c\mathbin{;}\beta_\mathbf{A} = \neg\varphi(p)$. Now $\neg \varphi(p) = {\sim}\varphi(p)$, so $(\alpha^{-1}_\mathbf{A}(y), x) \in -\varphi(p)= \alpha_\mathbf{A}\mathbin{;}\varphi(p)^{c\smile}$. This shows $(y, x) \in \varphi(p)^{c\smile}$, which implies $(x, y) \notin \varphi(p)$, a contradiction. 

Conversely, assume $(\beta_\mathbf{A}(y), \beta_{\mathbf{A}}(x)) \in \varphi(p)$ and suppose $(x, y) \notin  \varphi(p)$. The first part implies that $(\beta_\mathbf{A}(y), \beta_{\mathbf{A}}(x)) \in E_\mathbf{A}$, and so $(x, y) \in E_\mathbf{A}$. 
Hence, since $(x, y) \notin  \varphi(p)$, we have $(y, x) \in \varphi(p)^{c\smile}$. Now $(\alpha^{-1}_\mathbf{A}(y), y) \in \alpha_\mathbf{A}$, and therefore $(\alpha^{-1}_\mathbf{A}(y), x) \in \alpha_\mathbf{A}\mathbin{;}\varphi(p)^{c\smile} = -\varphi(p)$. But $\neg\varphi(p) = -\varphi(p)$, so $(\alpha^{-1}_\mathbf{A}(y), x) \in \neg\varphi(p)$. Consequently, $(\beta_\mathbf{A}(y), \beta_{\mathbf{A}}(x)) \in \varphi(p)^{c}$, and therefore $(\beta_\mathbf{A}(y), \beta_{\mathbf{A}}(x)) \notin \varphi(p)$, which is a contradiction.
\qed
\end{proof}

Using the previous two lemmas we can show that that $\alpha_{p\mathbf{A}p}$ and $\beta_{p\mathbf{A}p}$ are well-defined and satisfy the conditions of Theorem~\ref{Theorem:Dq(E)}. 

\begin{lemma}\label{lem:new_alpha_beta}
\begin{enumerate}[label=\textup{(\roman*)}]
\item The map $\alpha_{p\mathbf{A}p}: {X_\mathbf{A}/{\equiv}} \to {X_\mathbf{A}/{\equiv}}$ is a well-defined order automorphism of $\left(X_\mathbf{A}/{\equiv}, {\leqslant_{X_\mathbf{A}/{\equiv}}}\right)$ such that $\alpha_{p\mathbf{A}p}\subseteq E_{p\mathbf{A}p}$. 
\item The map $\beta_{p\mathbf{A}p}: {X_\mathbf{A}/{\equiv}} \to {X_\mathbf{A}/{\equiv}}$ is a well-defined self-inverse dual order automorphism of $\left(X_\mathbf{A}/{\equiv}, {\leqslant_{X_\mathbf{A}/{\equiv}}}\right)$ such that $\beta_{p\mathbf{A}p}\subseteq E_{p\mathbf{A}p}$.
\item $\alpha_{p\mathbf{A}p}\mathbin{;}\beta_{p\mathbf{A}p}\mathbin{;}\alpha_{p\mathbf{A}p} = \beta_{p\mathbf{A}p}$
\end{enumerate}
\end{lemma}

\begin{proof}
(i) Firstly, since $\alpha_\mathbf{A}$ is surjective, it follows that $\alpha_{p\mathbf{A}p}$ is 
surjective. 
To see that $\alpha_{p\mathbf{A}p}$ is well-defined, assume $[x]= [y]$. Then $x \equiv y$, and so $(x, y) \in \varphi(p)$ and $(y, x) \in \varphi(p)$. Hence, by Lemma~\ref{lem:alpha_order_aut_varphi(p)}, $(\alpha_\mathbf{A}(x), \alpha_\mathbf{A}(y)) \in \varphi(p)$ and $(\alpha_\mathbf{A}(y), \alpha_\mathbf{A}(x)) \in \varphi(p)$, which implies $\alpha_\mathbf{A}(x) \equiv \alpha_\mathbf{A}(y)$. This gives $[\alpha_\mathbf{A}(a)] = [\alpha_\mathbf{A}(y)]$, i.e., $\alpha_{p\mathbf{A}p}([x]) = \alpha_{p\mathbf{A}p}([y])$. 

Next we show that $\alpha_{p\mathbf{A}p}$ is an order automorphism of $\left(X_\mathbf{A}/{\equiv}, {\leqslant_{X_\mathbf{A}/{\equiv}}}\right)$. Let $[x], [y] \in X_\mathbf{A}/{\equiv}$. Then:
\begin{align*}
[x] \leqslant_{X_\mathbf{A}/{\equiv}} [y] & \textnormal{ iff } (x, y) \in \varphi(p)\\
& \textnormal{ iff } (\alpha_\mathbf{A}(x), \alpha_\mathbf{A}(y)) \in \varphi(p) & (\textnormal{by Lemma~\ref{lem:alpha_order_aut_varphi(p)}}) \\
& \textnormal{ iff } 
[\alpha_\mathbf{A}(x)]\leqslant_{X_\mathbf{A}/{\equiv}} [\alpha_\mathbf{A}(y)]\\
& \textnormal{ iff } 
\alpha_{p\mathbf{A}p}([x]) \leqslant_{X_\mathbf{A}/{\equiv}} \alpha_{p\mathbf{A}p}([y]).
\end{align*}
Finally, let $([x],  [y]) \in \alpha_{p\mathbf{A}p}$. Then we have $\alpha_{p\mathbf{A}p}([x]) = [y]$, and so $[\alpha_\mathbf{A}(x)] = [y]$. Hence, $\alpha_\mathbf{A}(x) \equiv y$, which implies that $(\alpha_\mathbf{A}(x), y) \in \varphi(p)$ and $(y, \alpha_\mathbf{A}(x)) \in \varphi(p)$.  Since $\varphi(p) \subseteq E_\mathbf{A}$, we have $(\alpha_\mathbf{A}(x), y) \in E_\mathbf{A}$. Now $(x, \alpha_\mathbf{A}(x)) \in \alpha_\mathbf{A}\subseteq E_\mathbf{A}$, so $(x, y) \in E_\mathbf{A}$, which shows that $([x], [y])\in E_{p\mathbf{A}p}$. 

(ii) Using Lemma~\ref{lem:beta_dual_order_aut_varphi(p)}, we can show  that $\beta_{p\mathbf{A}p}$ is a well-defined dual order automorphism. To show that $\beta_{p\mathbf{A}p}$ is self-inverse, let $[x]\in X_\mathbf{A}/{\equiv}$. Then, since $\beta_\mathbf{A}$ is self-inverse, we obtain $\beta_{p\mathbf{A}p}(\beta_{p\mathbf{A}p}([x])) = \beta_{p\mathbf{A}p}([\beta_\mathbf{A}(x)]) = [\beta_\mathbf{A}(\beta_\mathbf{A}(x))] = [x]$. 

(iii) We will show that $\alpha_{p\mathbf{A}p}\circ \beta_{p\mathbf{A}p}\circ\alpha_{p\mathbf{A}p} = \beta_{p\mathbf{A}p}$. Let $[x] \in X_\mathbf{A}/{\equiv}$. Then, since $\alpha_\mathbf{A}\mathbin{;}\beta_\mathbf{A}\mathbin{;}\alpha_\mathbf{A} = \beta_\mathbf{A}$, 
we get 
\begin{align*}
\alpha_{p\mathbf{A}p}(\beta_{p\mathbf{A}p}(\alpha_{p\mathbf{A}p}([x]))) & = \alpha_{p\mathbf{A}p}(\beta_{p\mathbf{A}p}([\alpha_{\mathbf{A}}(x)])) \\
& = \alpha_{p\mathbf{A}p}([\beta_{\mathbf{A}}(\alpha_{\mathbf{A}}(x))]) \\
& = [\alpha_\mathbf{A}(\beta_\mathbf{A}(\alpha_\mathbf{A}(x)))]\\
& = [\beta_\mathbf{A}(x)] = \beta_{p\mathbf{A}p}([x]).
\end{align*}
\qed
\end{proof}

An immediate consequence of the above is that the poset $(X/{\equiv}, \leqslant_{X_\mathbf{A}/\equiv})$, equipped with $E_{p\mathbf{A}p}$, $\alpha_{p\A p}$ and $\beta_{p \A p}$, can be used to build a DqRA as per  Theorem~\ref{Theorem:Dq(E)}. 


\begin{theorem}\label{thm:new_DqRA}
The algebra $\left\langle \mathsf{Up}\left(E_{p\mathbf{A}p}, \preccurlyeq_{E_{p\mathbf{A}p}}\right), \cap, \cup, \mathbin{;}, {\sim}, -, \neg, \leqslant_{X_{\mathbf{A}}/{\equiv}}\right\rangle$ is a DqRA. 
\end{theorem}

Our next goal is to show that the DqRA $p\mathbf{A}p$ can be embedded into the DqRA of binary relations $\left\langle \mathsf{Up}\left(E_{p\mathbf{A}p}, \preccurlyeq_{E_{p\mathbf{A}p}}\right), \cap, \cup, \mathbin{;}, {\sim}, -, \neg, \leqslant_{X_{\mathbf{A}}/{\equiv}}\right\rangle$. To this end, define $\psi: pAp \to \mathcal{P}\left((X_\mathbf{A}/{\equiv})^2\right)$ by setting, for all $a \in pAp$,
\[
\psi(a) := \{([x], [y]) \mid (x, y) \in \varphi(a)\}.
\]

\begin{lemma}\label{lem:psi_injective_to_Up}
The map $\psi$ is an injective function from $pAp$ to   
$\mathsf{Up}\left(E_{p\mathbf{A}p}, \preccurlyeq_{E_{p\mathbf{A}p}}\right)$.
\end{lemma}


\begin{proof}
We first show that $\psi(a)$ is an upset for all $a \in pAp$. 
Let $([x], [y])\in \psi(a)$ and $([u], [v]) \in E_{p\mathbf{A}p}$. Assume $([x], [y]) \preccurlyeq_{p\mathbf{A}p} ([u], [v])$. Since $([x], [y])\in \psi(a)$ and $([u], [v]) \in E_{p\mathbf{A}p}$, we have $(x, y) \in \varphi(a)$ and $(u, v)  \in E_\mathbf{A}$. We also have $[u] \leqslant_{X_\mathbf{A}/{\equiv}} [x]$ and $[y] \leqslant_{X_\mathbf{A}/{\equiv}} [v]$, and thus $(u, x) \in \varphi(p)$ and $(y, v) \in \varphi(p)$. Now $\varphi(p) \subseteq {\leqslant_{X_\mathbf{A}}}$, so $u\leqslant_{X_\mathbf{A}} x$ and $y \leqslant_{X_\mathbf{A}} v$. Hence, $(x, y) \preccurlyeq_{E_{\mathbf{A}}} (u, v)$, and therefore, since $\varphi(a)$ is an upset of $(E_\mathbf{A}, \preccurlyeq_{E_\mathbf{A}})$, we get $(u, v) \in \varphi(a)$, which shows $([u], [v]) \in \psi(a)$. 



Let $a, b \in pAp$ and assume $a \neq b$. Then $a, b \in A$, and so, since $\varphi$ is injective, it follows that $\varphi(a) \neq \varphi(b)$. This implies $\varphi(a) \not\subseteq \varphi(b)$ or $\varphi(b) \not\subseteq \varphi(a)$. Assume the first. Hence, there exist $x, y \in X_\mathbf{A}$ such that $(x, y) \in \varphi(a)$ but $(x, y) \notin \varphi(b)$. Consequently, $([x], [y]) \in \psi(a)$ and $([x], [y]) \notin \psi(b)$. This shows that $\psi(a) \neq \psi(b)$. Similarly, if $\varphi(b) \not\subseteq \varphi(a)$, then $\psi(a) \neq \psi(b)$. 
\qed
\end{proof}

In the next few lemmas we will show that $\psi$ is a homomorphism from $p\mathbf{A}p$ into $\mathbf{Dq}\left(E_{p\mathbf{A}p}, \preccurlyeq_{E_{p\mathbf{A}p}}\right)$. It is immediate from the definitions of $\psi$ and $\leqslant_{X_{\mathbf{A}}/{\equiv}}$ that $\psi(p) = {\leqslant_{X_{\mathbf{A}}/{\equiv}}}$.

\begin{lemma}\label{lem:psi_preserves_binary_ops}
The map $\psi: pAp \to \mathsf{Up}\left(E_{p\mathbf{A}p}, \preccurlyeq_{E_{p\mathbf{A}p}}\right)$ preserves the lattice operations and the monoid operation. 
\end{lemma}


\begin{proof}
The fact that $\psi(a\wedge b) = \psi(a)\cap\psi(b)$ and $\psi(a\vee b) = \psi(a)\cup\psi(b)$ for all $a,b\in pAp$
follows from basic set theory and the fact that $\varphi$ preserves $\wedge$ and $\vee$.



Let $a, b \in pAp$ and $x, y \in X_\mathbf{A}$. The first part implies $a, b \in A$. Hence, 
$([x], [y]) \in \psi(a\cdot b)$ iff $(x, y) \in \varphi_{\mathbf{A}}(a\cdot b)$
 iff $(x, y) \in \varphi_{\mathbf{A}}(a)\mathbin{;} \varphi(b)$
iff there exists  $z \in X$  such that $(x, z) \in \varphi(a)$ and $(z, y) \in \varphi(b)$ iff there exists $[z] \in X_\mathbf{A}/{\equiv}$ such that $([x], [z]) \in \psi(a)$  and $([z], [y]) \in \psi(b)$ 
 iff $([x], [y]) \in \psi(a)\mathbin{;}\psi(b)$.
\qed
\end{proof}


To show that $\psi$ preserves $\tw$ we will need the fact that 
$\alpha^{-1}_{p\mathbf{A}p}([x]) = [\alpha^{-1}_\mathbf{A}(x)]$ for all $[x] \in X_\mathbf{A}/{\equiv}$. This follows immediately from the definition of $\alpha_{p \A p}$. 


\begin{lemma}\label{lem:psi_preserves_unary_ops}
The map $\psi: pAp \to \mathsf{Up}\left(E_{p\mathbf{A}p}, \preccurlyeq_{E_{p\mathbf{A}p}}\right)$ preserves 
the unary operations $\tw$, $-$ and $\neg$. 
\end{lemma}

\begin{proof}
We will show that $\psi$ preserves $\tw$. The proof that it preserves $-$ follows similarly. Let $a \in pAp$ and $[x], [y] \in X_\mathbf{A}/{\equiv}$. 
Observe that $ ([x], [y]) \in \psi({\sim}a)$   iff $(x, y) \in \varphi({\sim}a)$ iff $(x, y) \in {\sim}\varphi(a)$.
Hence, 
\begin{align*}
([x], [y]) \in \psi({\sim}a) 
& \textnormal{ iff } (x, y) \in \varphi(a)^{c\smile}\mathbin{;}\alpha_\mathbf{A}\\
& \textnormal{ iff } (x,\alpha^{-1}_\mathbf{A}(y)) \in \varphi(a)^{c\smile}\\
& \textnormal{ iff } (\alpha^{-1}_\mathbf{A}(y), x) \in \varphi(a)^c\\
& \textnormal{ iff } (\alpha^{-1}_\mathbf{A}(y), x) \in E_\mathbf{A} \textnormal{ and } (\alpha^{-1}_\mathbf{A}(y), x) \notin \varphi(a)\\
& \textnormal{ iff } ([\alpha^{-1}_\mathbf{A}(y)], [x]) \in E_{p\mathbf{A}p} \textnormal{ and } ([\alpha^{-1}_\mathbf{A}(y)], [x]) \notin \psi(a)\\
& \textnormal{ iff } (\alpha^{-1}_{p\mathbf{A}p}([y]), [x]) \in E_{p\mathbf{A}p} \textnormal{ and } (\alpha^{-1}_{p\mathbf{A}p}([y]), [x]) \notin \psi(a)\\
& \textnormal{ iff }
([x], \alpha^{-1}_{p\mathbf{A}p}([y])) \in \psi(a)^{c\smile}\\
& \textnormal{ iff } ([x], [y]) \in \psi(a)^{c\smile}\mathbin{;}\alpha_{p\mathbf{A}p}\\
& \textnormal{ iff } ([x], [y]) \in {\sim}\psi(a)
\end{align*}



To show that $\psi$ preserves $\neg$, let $a \in pAp$ and $[x], [y] \in X_\mathbf{A}/{\equiv}$. Then we have $([x], [y]) \in \psi(\neg a)$ iff $(x, y) \in \varphi(\neg a)$ iff $(x, y) \in \neg \varphi(a)$ and hence,
\begin{align*}
& ([x], [y]) \in \psi(\neg a)\\
 \textnormal{ iff } &(x, y) \in \alpha_\mathbf{A}\mathbin{;}\beta_\mathbf{A}\mathbin{;}\varphi(a)^{c}\mathbin{;}\beta_\mathbf{A}\\
 \textnormal{ iff } &(\beta_\mathbf{A}(\alpha_\mathbf{A}(x)), \beta_\mathbf{A}(y)) \in \varphi(a)^{c}\\
 \textnormal{ iff } &(\beta_\mathbf{A}(\alpha_\mathbf{A}(x)), \beta_\mathbf{A}(y))  \in E_\mathbf{A} \textnormal{ and } (\beta_\mathbf{A}(\alpha_\mathbf{A}(x)), \beta_\mathbf{A}(y)) \notin \varphi(a)\\
 \textnormal{ iff } & ([\beta_\mathbf{A}(\alpha_\mathbf{A}(x))],[\beta_\mathbf{A}(y)]) \in E_{p\mathbf{A}p} \textnormal{ and } ([\beta_\mathbf{A}(\alpha_\mathbf{A}(x))],[\beta_\mathbf{A}(y)]) \notin \psi(a)\\
 \textnormal{ iff } &(\beta_{p\mathbf{A}p}(\alpha_{p\mathbf{A}p}([x]),\beta_{p\mathbf{A}p}([y])) \in E_{p\mathbf{A}p} \textnormal{ and } (\beta_{p\mathbf{A}p}(\alpha_{p\mathbf{A}p}([x]),\beta_{p\mathbf{A}p}([y]))  \notin \psi(a)\\
 \textnormal{ iff } & (\beta_{p\mathbf{A}p}(\alpha_{p\mathbf{A}p}([x]),\beta_{p\mathbf{A}p}([y])) \in \psi(a)^c\\
\textnormal{ iff } & ([x], [y]) \in \alpha_{p\mathbf{A}p}\mathbin{;}\beta_{p\mathbf{A}p}\mathbin{;}\psi(a)^{c}\mathbin{;}\beta_{p\mathbf{A}p}\\
\textnormal{ iff } & ([x], [y]) \in \neg\psi(a)
\end{align*}
\qed
\end{proof}

Combining Lemmas~\ref{lem:psi_injective_to_Up}, \ref{lem:psi_preserves_binary_ops}, and \ref{lem:psi_preserves_unary_ops} shows that the map $\psi$ is an embedding from $p \mathbf Ap$ into $\mathbf{Dq}\left(E_{p\mathbf{A}p}, \preccurlyeq_{E_{p\mathbf{A}p}}\right)$. 
It is clear from the construction that if $\mathbf{A}$ is a finitely representable DqRA, then $p\mathbf{A}p$ must also be finitely representable. 
We thus obtain the following result. 

\begin{theorem}\label{thm: pAp_representable}
Let $\mathbf{A} = \langle A,\wedge, \vee, \cdot, {\sim},{-},{\neg}, 1\rangle$ be a (finitely) representable DqRA, and let $p$ be a positive symmetric idempotent of $\mathbf{A}$. Then $p\mathbf{A}p$ is (finitely) representable. 
\end{theorem}






Returning to 
Example~\ref{Example:Consturction1}, the theorem above shows that the posets with their respective $E$, $\alpha$ and $\beta$ in Figure~\ref{fig:Construction1} will give representations of the contractions of $\mathbf{A}=D^6_{3,5,2}$. Note that the algebras $\mathbf{Dq}(\mathbf{X}/{\equiv_a}, \preccurlyeq_{\equiv_a})$ and 
$\mathbf{Dq}(\mathbf{X}/{\equiv_b}, \preccurlyeq_{\equiv_b})$ are not isomorphic because $\beta_{a\A a} \neq \beta_{b \A b}$, but they both contain a subalgebra isomorphic to $\mathbf{S}_3$ (since   $\mathbf{S}_3 \cong a\mathbf{A} a \cong b\mathbf{A} b$). 

\section{DqRAs that are not finitely representable}\label{sec:not-finite}

It has been shown that under certain conditions a DqRA cannot be represented using a finite poset $(X,\leqslant)$. 
\begin{theorem}{\normalfont \cite[Theorem 5.12]{RDqRA25}}\label{Theorem:NoFiniteRep}
Let $\A = \langle A, \wedge, \vee, \cdot, \sim, -,\neg,1\rangle$ be a DqRA. If there exists $a\in A$ such that $0<a< 1$ and $a^2 \leqslant 0$, then
$\A$ is not finitely representable.
\end{theorem}

\begin{example}
Each of the DqRAs depicted in Figure~\ref{Figure:DqRAChainsNotFinRep}
contains an element $a$ (also labelled $a$ in their depictions) that satisfies the
criteria in Theorem~\ref{Theorem:NoFiniteRep}, i.e.,
$0<a<1$ and $a^2\leqslant 0$.  Hence, these
DqRAs are not finitely representable.
\end{example}

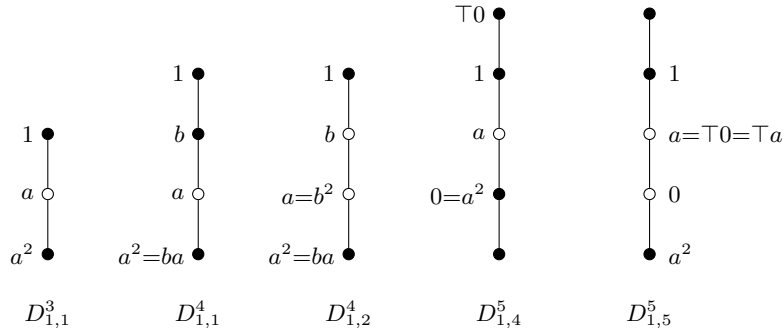
\begin{figure}
    \centering
\begin{tikzpicture}
\begin{scope}[xshift=-4cm,scale=0.8]
\node[draw,circle,inner sep=1.5pt,fill] (bot) at (0,0) {};
\node[draw,circle,inner sep=1.5pt] (a) at (0,1) {};
\node[draw,circle,inner sep=1.5pt,fill] (top) at (0,2) {};
\path [-] (bot) edge node {} (a);
\path [-] (a) edge node {} (top);
\node[label,anchor=east,xshift=1pt] at (top) {$1$};
\node[label,anchor=east,xshift=1pt] at (a) {$a$};
\node[label,anchor=east,xshift=1pt] at (bot) {$a^2$};

\node[] at (0,-1) {$D^3_{1,1}$};
\end{scope}

\begin{scope}[xshift=-2cm,scale=0.8]
\node[draw,circle,inner sep=1.5pt,fill] (bot) at (0,0) {};
\node[draw,circle,inner sep=1.5pt] (b) at (0,1) {};
\node[draw,circle,inner sep=1.5pt,fill] (a) at (0,2) {};
\node[draw,circle,inner sep=1.5pt,fill] (top) at (0,3) {};
\path [-] (bot) edge node {} (b);
\path [-] (b) edge node {} (a);
\path [-] (a) edge node {} (top);
\node[label,anchor=east,xshift=1pt] at (top) {$1$};
\node[label,anchor=east,xshift=1pt] at (a) {$b$};
\node[label,anchor=east,xshift=1pt] at (b) {$a$};
\node[label,anchor=east,xshift=1pt] at (bot) {$a^2{=}ba$};

\node[] at (0,-1) {$D^4_{1,1}$};
\end{scope}

\begin{scope}[xshift=0cm,scale=0.8]
\node[draw,circle,inner sep=1.5pt,fill] (bot) at (0,0) {};
\node[draw,circle,inner sep=1.5pt] (b) at (0,1) {};
\node[draw,circle,inner sep=1.5pt] (a) at (0,2) {};
\node[draw,circle,inner sep=1.5pt,fill] (top) at (0,3) {};
\path [-] (bot) edge node {} (b);
\path [-] (b) edge node {} (a);
\path [-] (a) edge node {} (top);
\node[label,anchor=east,xshift=1pt] at (top) {$1$};
\node[label,anchor=east,xshift=1pt] at (a) {$b$};
\node[label,anchor=east,xshift=1pt] at (b) {$a{=}b^2$};
\node[label,anchor=east,xshift=1pt] at (bot) {$a^2{=}ba$};

\node[] at (0,-1) {$D^4_{1,2}$};
\end{scope}

\begin{scope}[xshift=2cm,scale=0.8]
\node[draw,circle,inner sep=1.5pt,fill] (bot) at (0,0) {};
\node[draw,circle,inner sep=1.5pt,fill] (0) at (0,1) {};
\node[draw,circle,inner sep=1.5pt] (a) at (0,2) {};
\node[draw,circle,inner sep=1.5pt,fill] (1) at (0,3) {};
\node[draw,circle,inner sep=1.5pt,fill] (top) at (0,4) {};
\path [-] (bot) edge node {} (0);
\path [-] (0) edge node {} (a);
\path [-] (a) edge node {} (1);
\path [-] (1) edge node {} (top);
\node[label,anchor=east,xshift=1pt] at (top) {$\top0$};
\node[label,anchor=east,xshift=1pt] at (1) {$1$};
\node[label,anchor=east,xshift=1pt] at (a) {$a$};
\node[label,anchor=east,xshift=1pt] at (0) {$0{=}a^2$};

\node[] at (0,-1) {$D^5_{1,4}$};
\end{scope}

\begin{scope}[xshift=4cm,scale=0.8]
\node[draw,circle,inner sep=1.5pt,fill] (bot) at (0,0) {};
\node[draw,circle,inner sep=1.5pt] (0) at (0,1) {};
\node[draw,circle,inner sep=1.5pt] (a) at (0,2) {};
\node[draw,circle,inner sep=1.5pt,fill] (1) at (0,3) {};
\node[draw,circle,inner sep=1.5pt,fill] (top) at (0,4) {};
\path [-] (bot) edge node {} (0);
\path [-] (0) edge node {} (a);
\path [-] (a) edge node {} (1);
\path [-] (1) edge node {} (top);
\node[label,anchor=west,xshift=1pt] at (1) {$1$};
\node[label,anchor=west,xshift=1pt] at (a) {$a{=}\top0{=}\top a$};
\node[label,anchor=west,xshift=1pt] at (0) {$0$};
\node[label,anchor=west,xshift=1pt] at (bot) {$a^2$};

\node[] at (0,-1) {$D^5_{1,5}$};
\end{scope}

\end{tikzpicture}
    \caption{DqRAs that are not finitely representable due to Theorem~\ref{Theorem:NoFiniteRep}.}
    \label{Figure:DqRAChainsNotFinRep}
\end{figure}

The contrapositive of Theorem~\ref{thm: pAp_representable} states that if 
$\mathbf{A}=\langle A,\wedge,\vee,\cdot,
\sim, {-},\neg,1\rangle$ is a DqRA, $p$ a
positive symmetric idempotent of $\mathbf{A}$ and $p\mathbf{A}p$ is not
finitely representable, then $\mathbf{A}$ 
is not finitely representable. This can be used in combination with Theorem~\ref{Theorem:NoFiniteRep} 
to find DqRAs that are not finitely
representable even though they do not satisfy the conditions laid out in Theorem~\ref{Theorem:NoFiniteRep}.

\begin{example}
Table~\ref{Table:NotFinRep} lists
small DqRAs that do not meet the
conditions of Theorem~\ref{Theorem:NoFiniteRep}.  For
each of these DqRAs there is
no element $a$ such that $0<a<1$ since
$0\nleqslant 1$. 
However, as indicated in the
table, they each have
at least one positive symmetric
idempotent element $p$ such that 
the 
contraction 
$p\A p$
is a DqRA that
is not finitely representable due to
Theorem~\ref{Theorem:NoFiniteRep}. 
Hence, by Theorem~\ref{thm: pAp_representable} none of the DqRAs 
in Table~\ref{Table:NotFinRep} are finitely 
representable.
\end{example}

    \begin{table}[ht!]
    \centering
    \begin{tblr}{|Q|Q[2.5cm,valign=h]|Q|}\hline
    The DqRA $\mathbf{A}$
    &Positive symmetric idempotent
    &$p\mathbf{A}p$\\ \hline    
    \begin{tikzpicture}
           \begin{scope}[xshift=4.3cm,scale=0.55]
    \node[draw,circle,inner sep=1.5pt,fill] (bot) at (0,0) {};
    \node[draw,circle,inner sep=1.5pt,fill] (1) at (-1,1) {};
    \node[draw,circle,inner sep=1.5pt] (a) at (1,1) {};
    \node[draw,circle,inner sep=1.5pt,fill] (top) at (0,2) {};
    \path [-] (bot) edge node {} (1);
    \path [-] (bot) edge node {} (a);
    \path [-] (1) edge node {} (top);
    \path [-] (a) edge node {} (top);
    \node[label,anchor=east,xshift=1pt] at (top) {$\top$};
    \node[label,anchor=east,xshift=1pt] at (1) {$1{=}0$};
    \node[label,anchor=west,xshift=1pt] at (a) {$a{=}\top a$};
    \node[label,anchor=east,xshift=1pt] at (bot) {$\bot{=}a^2$};
    \node[] at (0,-1) {$D^4_{3,1}$};
    \end{scope} 
    \end{tikzpicture}
           &
    $p=\top$
        &
    \begin{tikzpicture}
        \begin{scope}[xshift=0cm,scale=0.55]
    \node[draw,circle,inner sep=1.5pt,fill] (bot) at (0,0) {};
    \node[draw,circle,inner sep=1.5pt] (a) at (0,1) {};
    \node[draw,circle,inner sep=1.5pt,fill] (top) at (0,2) {};
    \path [-] (bot) edge node {} (a);
    \path [-] (a) edge node {} (top);
    \node[label,anchor=east,xshift=1pt] at (top) {$1_p=\top$};
    \node[label,anchor=east,xshift=1pt] at (a) {$a$};
    \node[label,anchor=east,xshift=1pt] at (bot) {$0_p=a^2$};
    
    \node[] at (0,-1) {$pAp\cong D^3_{1,1}$};
    \end{scope}
    \end{tikzpicture}\\
    \hline
    \begin{tikzpicture}
           \begin{scope}[xshift=0cm,scale=0.55]
    \node[draw,circle,inner sep=1.5pt,fill] (bot) at (0,0) {};
    \node[draw,circle,inner sep=1.5pt,fill] (c) at (0,1) {};
    \node[draw,circle,inner sep=1.5pt,fill] (1) at (-1,2) {};
    \node[draw,circle,inner sep=1.5pt] (b) at (1,2) {};
     \node[draw,circle,inner sep=1.5pt,fill] (a) at (0,3) {};
    \node[draw,circle,inner sep=1.5pt,fill] (top) at (0,4) {};
    \path [-] (bot) edge node {} (c);
    \path [-] (c) edge node {} (1);
    \path [-] (c) edge node {} (b);
    \path [-] (1) edge node {} (a);
    \path [-] (b) edge node {} (a);
    \path [-] (a) edge node {} (top);
    \node[label,anchor=east,xshift=1pt] at (top) {$\top c$};
    \node[label,anchor=east,xshift=1pt] at (a) {$a$};
    \node[label,anchor=east,xshift=1pt] at (1) {$1{=}0$};
    \node[label,anchor=west,xshift=1pt] at (b) {$b{=}ab$};
    \node[label,anchor=east,xshift=1pt] at (c) {$c{=}bc{=}b^2{=}ac$};
    \node[] at (0,-1) {$D^6_{3,2}$};
    \end{scope} 
    \end{tikzpicture}
        &
    $p=a$
        &
    \begin{tikzpicture}
        \begin{scope}[xshift=0cm,scale=0.55]
    \node[draw,circle,inner sep=1.5pt,fill] (bot) at (0,0) {};
    \node[draw,circle,inner sep=1.5pt,fill] (0) at (0,1) {};
    \node[draw,circle,inner sep=1.5pt] (a) at (0,2) {};
    \node[draw,circle,inner sep=1.5pt,fill] (1) at (0,3) {};
    \node[draw,circle,inner sep=1.5pt,fill] (top) at (0,4) {};
    \path [-] (bot) edge node {} (0);
    \path [-] (0) edge node {} (a);
    \path [-] (a) edge node {} (1);
    \path [-] (1) edge node {} (top);
    \node[label,anchor=east,xshift=1pt] at (top) {$\top c=\top$};
    \node[label,anchor=east,xshift=1pt] at (1) {$1_p=a$};
    \node[label,anchor=east,xshift=1pt] at (a) {$b$};
    \node[label,anchor=east,xshift=1pt] at (0) {$0_p{=}b^2=c$};
    
    \node[] at (0,-1) {$pAp\cong D^5_{1,4}$};
        \end{scope}
    \end{tikzpicture}\\
        \hline
    \begin{tikzpicture}
           \begin{scope}[xshift=0cm,scale=0.55]
    \node[draw,circle,inner sep=1.5pt,fill] (bot) at (0,0) {};
    \node[draw,circle,inner sep=1.5pt] (c) at (0,1) {};
    \node[draw,circle,inner sep=1.5pt,fill] (1) at (-1,2) {};
    \node[draw,circle,inner sep=1.5pt] (b) at (1,2) {};
     \node[draw,circle,inner sep=1.5pt,fill] (a) at (0,3) {};
    \node[draw,circle,inner sep=1.5pt,fill] (top) at (0,4) {};
    \path [-] (bot) edge node {} (c);
    \path [-] (c) edge node {} (1);
    \path [-] (c) edge node {} (b);
    \path [-] (1) edge node {} (a);
    \path [-] (b) edge node {} (a);
    \path [-] (a) edge node {} (top);
    \node[label,anchor=east,xshift=1pt] at (a) {$a$};
    \node[label,anchor=east,xshift=1pt] at (1) {$1{=}0$};
    \node[label,anchor=west,xshift=1pt] at (b) {$b{=}ab{=}\top a{=}\top b$};
    \node[label,anchor=east,xshift=1pt] at (c) {$c{=}ac$};
    \node[label,anchor=east,xshift=1pt] at (bot) {$b^2$};
    \node[] at (0,-1) {$D^6_{3,4}$};
    \node[] at (0,4.3) {$\;$};
    \end{scope} 
    \end{tikzpicture}
        &
    $p=a$ (or use $p=\top$ to
    obtain $D^3_{1,1}$)
        &
    \begin{tikzpicture}
        \begin{scope}[xshift=4cm,scale=0.55]
    \node[draw,circle,inner sep=1.5pt,fill] (bot) at (0,0) {};
    \node[draw,circle,inner sep=1.5pt] (0) at (0,1) {};
    \node[draw,circle,inner sep=1.5pt] (a) at (0,2) {};
    \node[draw,circle,inner sep=1.5pt,fill] (1) at (0,3) {};
    \node[draw,circle,inner sep=1.5pt,fill] (top) at (0,4) {};
    \path [-] (bot) edge node {} (0);
    \path [-] (0) edge node {} (a);
    \path [-] (a) edge node {} (1);
    \path [-] (1) edge node {} (top);
    \node[label,anchor=east,xshift=1pt] at (1) {$1_p=a$};
    \node[label,anchor=east,xshift=1pt] at (a) {$b$};
    \node[label,anchor=east,xshift=1pt] at (0) {$0_p=c$};
    \node[label,anchor=east,xshift=1pt] at (bot) {$b^2=\bot$};
    
    \node[] at (0,-1) {$pAp\cong D^5_{1,5}$};
    \end{scope}
    \end{tikzpicture}\\
    \hline
        \begin{tikzpicture}
           \begin{scope}[xshift=0cm,scale=0.55]
    \node[draw,circle,inner sep=1.5pt,fill] (bot) at (-1,0) {};
    \node[draw,circle,inner sep=1.5pt] (0) at (-2,1) {};
    \node[draw,circle,inner sep=1.5pt] (b) at (0,1) {};
    \node[draw,circle,inner sep=1.5pt,fill] (a) at (-1,2) {};
     \node[draw,circle,inner sep=1.5pt, fill] (1) at (1,2) {};
    \node[draw,circle,inner sep=1.5pt,fill] (top) at (0,3) {};
    \path [-] (bot) edge node {} (0);
    \path [-] (bot) edge node {} (b);
    \path [-] (0) edge node {} (a);
    \path [-] (b) edge node {} (a);
    \path [-] (b) edge node {} (1);
    \path [-] (1) edge node {} (top);
    \path [-] (a) edge node {} (top);   
    \node[label,anchor=east,xshift=1pt] at (a) {$a{=}0^2{=}\top a$};
    \node[label,anchor=west,xshift=1pt] at (1) {$1$};
    \node[label,anchor=east,xshift=1pt] at (0) {$0$};
    \node[label,anchor=west,xshift=1pt] at (b) {$b{=}\top b$};
    \node[label,anchor=east,xshift=1pt] at (bot) {$ab$};
    \node[] at (0,-1) {$D^6_{4,3}$};
    \end{scope} 
    \end{tikzpicture}
        &
    $p=\top$
        &
        \begin{tikzpicture}
    \begin{scope}[xshift=0cm,scale=0.55]
    \node[draw,circle,inner sep=1.5pt,fill] (bot) at (0,0) {};
    \node[draw,circle,inner sep=1.5pt] (b) at (0,1) {};
    \node[draw,circle,inner sep=1.5pt,fill] (a) at (0,2) {};
    \node[draw,circle,inner sep=1.5pt,fill] (top) at (0,3) {};
    \path [-] (bot) edge node {} (b);
    \path [-] (b) edge node {} (a);
    \path [-] (a) edge node {} (top);
    \node[label,anchor=east,xshift=1pt] at (top) {$1_p=\top$};
    \node[label,anchor=east,xshift=1pt] at (a) {$a$};
    \node[label,anchor=east,xshift=1pt] at (b) {$b$};
    \node[label,anchor=east,xshift=1pt] at (bot) {$0_p=b^2=ab$};
    
    \node[] at (0,-1) {$pAp\cong D^4_{1,1}$};
    \end{scope}
    \end{tikzpicture}\\
    \hline  
    \begin{tikzpicture}
           \begin{scope}[xshift=0cm,scale=0.55]
    \node[draw,circle,inner sep=1.5pt,fill] (bot) at (-1,0) {};
    \node[draw,circle,inner sep=1.5pt] (0) at (-2,1) {};
    \node[draw,circle,inner sep=1.5pt] (b) at (0,1) {};
    \node[draw,circle,inner sep=1.5pt] (a) at (-1,2) {};
     \node[draw,circle,inner sep=1.5pt,fill] (1) at (1,2) {};
    \node[draw,circle,inner sep=1.5pt,fill] (top) at (0,3) {};
    \path [-] (bot) edge node {} (0);
    \path [-] (bot) edge node {} (b);
    \path [-] (0) edge node {} (a);
    \path [-] (b) edge node {} (a);
    \path [-] (b) edge node {} (1);
    \path [-] (1) edge node {} (top);
    \path [-] (a) edge node {} (top);   
    \node[label,anchor=east,xshift=1pt] at (a) {$a{=}\top 0{=}\top a$};
    \node[label,anchor=west,xshift=1pt] at (1) {$1$};
    \node[label,anchor=east,xshift=1pt] at (0) {$0$};
    \node[label,anchor=west,xshift=1pt] at (b) {$b{=}0^2{=}a^2{=}\top b$};
    \node[label,anchor=east,xshift=1pt] at (bot) {$ab$};
    \node[] at (0,-1) {$D^6_{4,4}$};
    \end{scope} 
    \end{tikzpicture}
        &
    $p=\top$
        &
        \begin{tikzpicture}
    \begin{scope}[xshift=0cm,scale=0.55]
    \node[draw,circle,inner sep=1.5pt,fill] (bot) at (0,0) {};
    \node[draw,circle,inner sep=1.5pt] (b) at (0,1) {};
    \node[draw,circle,inner sep=1.5pt] (a) at (0,2) {};
    \node[draw,circle,inner sep=1.5pt,fill] (top) at (0,3) {};
    \path [-] (bot) edge node {} (b);
    \path [-] (b) edge node {} (a);
    \path [-] (a) edge node {} (top);
    \node[label,anchor=east,xshift=1pt] at (top) {$1_p=\top$};
    \node[label,anchor=east,xshift=1pt] at (a) {$a$};
    \node[label,anchor=east,xshift=1pt] at (b) {$a^2{=}b$};
    \node[label,anchor=east,xshift=1pt] at (bot) {$b^2{=}ab$};
    
    \node[] at (0,-1) {$pAp\cong D^4_{1,2}$};
    \end{scope}
    \end{tikzpicture}\\
    \hline
    \end{tblr}    
    \caption{DqRAs not finitely representable due to their contractions.}
    \label{Table:NotFinRep}
    \end{table}

Using the examples from Table~\ref{Table:NotFinRep}, we are able to formulate and prove a stronger theorem which gives more general conditions under which a DqRA $\mathbf{A}$ will not be finitely representable.  

\begin{theorem}\label{Theorem:NewNoFiniteRep}
Let $\mathbf{A} = \langle A, \wedge, \vee, \cdot, \sim, -,\neg,1\rangle$ be a DqRA. If $p$ is a positive symmetric idempotent of $\mathbf A$ and there exists $b\in A$ such that $p\,b=b=b\,p$, $-p< b < p$ and $b^2\leqslant -p$,
then $A$ is not finitely representable.
\end{theorem}

\begin{proof}
Let $\mathbf{A}$ be a distributive quasi relation algebra with $p\in A$ a positive
symmetric idempotent and $b\in A$ such that
$b\,p=b=p\,b$, 
$-p< b < p$ and $b^2\leqslant -p$.
Suppose to the contrary that $\mathbf{A}$ is
representable over $(X,\leqslant)$ and $E$ with $|\alpha|=n\in\mathbb{Z}^+$, i.e., 
$\alpha^n(x)=x$ for all $x \in X$.  Then
there exists an embedding $\varphi:A\to  \mathsf{Up}(\mathbf{E})$ such that 
$\varphi(1)={\leqslant}$.

Moreover, there exists $(x,y)\in E$ such
that $(x,y)\in\varphi(b)$ but $(x,y)\notin\varphi(-p)$.
Therefore, 
    
$\begin{array}{lll}
(x,y)\notin -\varphi(p) &\textnormal{ iff }
(x,y)\notin \alpha;\left(\varphi(p)^c\right)^\smile\\
&\textnormal{ iff } (x,y)\notin \left(\alpha;\varphi(p)^\smile\right)^c\\
&\textnormal{ iff } (x,y)\in \alpha;\varphi(p)^\smile\\
&\textnormal{ iff } (\alpha(x),y)\in\varphi(p)^\smile\\
&\textnormal{ iff } (y,\alpha(x))\in\varphi(p).\\
    \end{array}$

As $(x,y)\in\varphi(b)$ and $(y,\alpha(x))\in\varphi(p)$, it
follows that $(x,\alpha(x))\in\varphi(b);\varphi(p)
=\varphi(b\,p)=\varphi(b)$ since $b\,p=b$
by assumption. 
    
Next, 
$(x,y)\in\varphi(b) \subseteq \varphi(p)$, so since $(y,\alpha(x))\in\varphi(p)$ and $p^2 = p$, we get
$(x,\alpha(x))\in \varphi(p)\mathbin{;}\varphi(p) = \varphi(p)$. Hence, by Lemma~\ref{lem:alpha_order_aut_varphi(p)}, $(\alpha(x),\alpha(\alpha(x))) \in \varphi(p)$. 
But $(x,\alpha(x))\in\varphi(b)$, so $(x,\alpha(\alpha(x)))\in\varphi(b);\varphi(p)=\varphi(b\,p)=\varphi(b)$. 
Repeated application of Lemma~\ref{lem:alpha_order_aut_varphi(p)} and the composition of successive pairs obtained in this way then gives $(x,\alpha^n(x))=(x,x)\in\varphi(b)$ since $|\alpha|=n$. 

Finally, combining $(x,x)\in\varphi(b)$
and $(x,y)\in\varphi(b)$ gives   $(x,y)\in\varphi(b);\varphi(b)=\varphi(b^2)\subseteq \varphi(-p)$, contradicting that $(x,y)\notin \varphi(-p)$. \qed 
\end{proof}

\begin{example}
    Consider $D^4_{3,1}$, depicted in the first row of
    Table~\ref{Table:NotFinRep}. Here $\top$ is a positive symmetric idempotent and the element $a$
    satisfies $\top a=a=a\top$ and $-\top=\bot<a<\top$ and $a^2=\bot\leqslant -\top$.  Thus, by
    Theorem~\ref{Theorem:NewNoFiniteRep}, $D^4_{3,1}$
    is not finitely representable.
\end{example}



\end{document}